\documentclass[sn-mathphys-num]{sn-jnl}
\usepackage{lineno,hyperref}
\usepackage{url}
\usepackage{stmaryrd}
\usepackage{tikz-cd}
\usepackage{yfonts}
\usepackage[utf8]{inputenc}
\usepackage[T1]{fontenc}
\usepackage [autostyle]{csquotes}
\usepackage{caption}
\usepackage{amsfonts}
\usepackage{amsthm}
\usepackage{eqnarray}
\usepackage{mathtools}

\usepackage{float}
\usepackage{amssymb}
\usepackage{amsmath}
\usepackage{enumerate}
\usepackage{lipsum}
\theoremstyle{thmstyleone}
\newtheorem{thm}{Theorem}
\newtheorem{lem}{Lemma}
\newtheorem{prop}{Proposition}
\newtheorem{cor}{Corollary}
\newtheorem{defn}{Definition}
\newtheorem{exmp}{Example}

\newtheorem{rem}{Remark}
\newtheorem{note}{Note}

\begin{document}
	\title{Positive Instantial Neighbourhood logic: Typed Completeness and Admissible-Open Representation}{}
	\author*[1,2]{\fnm{Litan Kumar} \sur{Das}}\email{ld06iitkgp@gmail.com}
	\author[1,2]{\fnm{Anupam } \sur{Khanra}}\email{anupamk.math.rs@jadavpuruniversity.in}
	\equalcont{These authors contributed equally to this work.}
	\author[1,2]{\fnm{Sujit Kumar} \sur{Sardar}}\email{sujitk.sardar@jadavpuruniversity.in}
	\equalcont{These authors contributed equally to this work.}
	\affil*[1]{\orgdiv{Department of Mathematics}, \orgname{Mahishadal Raj College}, \orgaddress{\street{Mahishadal}, \city{}, \postcode{721628}, \state{West Bengal}, \country{India}}}
	
	\affil[2]{\orgdiv{Department of Mathematics}, \orgname{Jadavpur University}, \orgaddress{\street{Jadavpur}, \city{Kolkata}, \postcode{700032}, \state{West Bengal}, \country{India}}}
	
	\affil[3]{\orgdiv{Department of Mathematics}, \orgname{Jadavpur University}, \orgaddress{\street{Jadavpur}, \city{Kolkata}, \postcode{700032}, \state{West Bengal}, \country{India}}}
	
		
\abstract{Instantial neighbourhood logic is a modal language for neighbourhood frames in which formulas can express information about the kinds of worlds occurring inside a neighbourhood of a given world. In this paper, we study a positive, negation- and implication-free version of instantial neighbourhood logic with two primitive instantial modalities, one of \(\Box\)-type and one of \(\Diamond\)-type. Since classical negation is not available, the two modalities are treated independently. We introduce the language and proof system of positive instantial neighbourhood logic (PINL) and interpret it over persistent two-sided neighbourhood models. We then define a typed persistent neighbourhood semantics, used as an auxiliary canonical semantics to control witness and co-witness conditions. This yields a truth lemma and a typed completeness theorem for PINL. On the algebraic side, we introduce \(2\)-$\mathrm{DLIO}$s, bounded distributive lattices equipped with two families of instantial operations, as the algebraic semantics of PINL. We prove algebraic soundness and completeness via the Lindenbaum \(2\)-$\mathrm{DLIO}$. Finally, we construct the canonical bitopological PINL-space and show that the algebra of its admissible positive opens is isomorphic to the Lindenbaum \(2\)-$\mathrm{DLIO}$. Thus the paper establishes a canonical admissible-open representation of positive instantial neighbourhood logic, providing a first step toward a future duality theory.}

	\maketitle
\section{Introduction}\label{INT}
Positive logic is the negation-and implication free fragment of propositional logic, constructed from propositional variables, the constants \(\top\) and \(\bot\), and the lattice connectives $\wedge$ and $\vee$. It may also be viewed as the finitary part of geometric logic \cite{vickers1989topology}, where finite conjunctions and finite disjunctions are allowed. Algebraically, this positive base is naturally interpreted in bounded distributive lattices. Positive modal logic, initiated by Dunn \cite{dunn1995positive}, extends positive propositional logic by adding two normal modal operators, $\Box$ and $\Diamond$, while still avoiding classical negation. Subsequent research produced relational, algebraic, and duality-theoretic semantics for positive modal logics, including Priestley-style dualities and coalgebraic approaches; see, for example, \cite{celani1997new,celani1999priestley,celani2012note,gehrke2005sahlqvist,kikot2018strictly,palmigiano2004coalgebraic,sadrzadeh2010positive}). More recently, de Groot studied positive monotone modal logic \cite{de2021positive} as the positive fragment of monotone modal logic and proved a duality between suitable neighbourhood spaces and distributive lattices with monotone operators. These works show that positive modal logics are naturally connected with ordered spaces, distributive lattices, and coalgebraic duality.\\
Neighbourhood semantics for modal logic has its roots in the work of Scott and Montague and was later systematically developed in the study of non-normal modal logics; see, for example, \cite{scott1970advice,montague1970universal,chellas1980modal,hansen2009neighbourhood,pacuit2017neighborhood}. 
Instantial neighbourhood logic was introduced in \cite{van2017instantial} as a modal language for neighbourhood frames. Its modalities can express not only that a neighbourhood satisfies a scope condition, but also that it contains points satisfying specified instance formulas. A coalgebraic duality for INL was later developed in \cite{bezhanishvili2020duality}. In that work, descriptive INL-frames were shown to be dual to Boolean algebras with instantial operators($\mathrm{BAIO}$). The same paper also points to positive INL as a natural related formalism, where Boolean algebras are replaced by distributive lattices with instantial operators. It further suggests a geometric version of INL, guided by the slogan that geometric INL should arise from positive INL by adding Scott-continuity. This connects the present direction with coalgebraic geometric logic and with point-free topological approaches to modal logic \cite{bezhanishvili2022coalgebraic,vickers2004double,vickers1989topology}.\\
The above literature suggests a positive direction for INL, but it does not develop a full proof-theoretic and canonical semantic treatment of such a logic. In particular, what is still missing is a positive consequence system with two independent instantial modalities, a persistent two-sided neighbourhood semantics, a canonical completeness argument adapted to the absence of negation, and a representation of the resulting algebra by admissible opens of a canonical bitopological space. The present paper addresses this more specific gap.

We develop positive instantial neighbourhood logic, abbreviated PINL. The language has two primitive instantial modalities: one of $\Box$-type and one of $\Diamond$-type. Since negation is not available, the two modalities are not defined from one another. Instead, they are treated independently. This leads to a two-sided semantic and algebraic setting: the $\Box$-type modality is governed by witness conditions, while the $\Diamond$-type modality is governed by co-witness conditions.\\
The typed semantics used in this paper is an auxiliary canonical semantics. The typed semantics is used as an auxiliary tool for the canonical construction; it is not intended to be a separate final semantics for PINL. In particular, formula labels keep track of which modal formula a finite neighbourhood set is intended to witness or to refute as a co-witness. This prevents an extensional neighbourhood set chosen for one modal formula from being incorrectly used for another modal formula. The typed construction is then converted, in the final representation section, into an admissible-open representation.\\
The proposed logic is useful for reasoning about local positive information. In settings such as information states or observational models, one may want to say that a neighbourhood satisfies a general condition and also contains specified positive instances, without relying on negation or implication. The two primitive instantial modalities allow us to distinguish witness behaviour from its dual co-witness behaviour, while the distributive-lattice base keeps the logic compatible with order and persistence.\\
The main contributions of the paper are the following. First, we formulate a positive two-sided version of instantial neighbourhood logic with two primitive modalities, \(\Box\) and \(\Diamond\), and give its persistent neighbourhood semantics. Second, we show why the ordinary unlabelled neighbourhood semantics is not adequate for a direct canonical truth lemma, and we introduce a typed persistent neighbourhood semantics to control the witness and co-witness conditions. This yields the typed truth lemma and the corresponding completeness theorem. Third, we introduce \(2\)-\(\mathrm{DLIO}\)s as the algebraic counterpart of the independent \(\Box\)- and \(\Diamond\)-fragments. The Lindenbaum-algebra argument gives the expected algebraic soundness and completeness. Finally, we replace formula-labelled canonical neighbourhoods by admissible-open labels and construct the canonical bitopological \(\mathrm{PINL}\)-space. We prove that the algebra of admissible positive opens of this space is isomorphic to the Lindenbaum \(2\)-\(\mathrm{DLIO}\). This gives a canonical admissible-open representation theorem for \(\mathrm{PINL}\).

The scope of the paper is intentionally limited. We prove a canonical representation theorem, not a categorical duality theorem. A Priestley-style or Vietoris-style duality for $2$-$\mathrm{DLIO}$s would require a suitable category of descriptive PINL-spaces, appropriate morphisms, and a functorial treatment of the instantial neighbourhood structure. This is left for future work.\\
The paper is organised as follows. \\
Section \ref{PRE} recalls the basic order-theoretic and lattice-theoretic notions used in the paper, including posets, upward closed sets, and distributive lattices. It also reviews the classical language of instantial neighbourhood logic. Section \ref{PNSS} introduces the language of PINL and discusses the persistence of truth. Section \ref{PSPINL} presents the proof system of PINL, including the propositional axioms, modal axioms, and the main derivable principles used later. Section \ref{SCPINL} proves soundness of PINL with respect to persistent two-sided neighbourhood models. Section \ref{CPMTC} introduces typed persistent neighbourhood models. These models are used as an auxiliary semantics for the canonical construction. In this section we construct the canonical typed model and prove the truth lemma, completeness, and the separation result for PINL. Section \ref{ASOP} introduces \(2\)-$\mathrm{DLIO}$s and formulates the algebraic semantics of PINL. It also develops the Lindenbaum algebra and proves algebraic soundness and completeness. Section \ref{BPAOR} develops the bitopological view of PINL. It introduces bitopological PINL-spaces, constructs the canonical bitopological PINL-space, and proves the admissible-open representation theorem for the Lindenbaum \(2\)-$\mathrm{DLIO}$ of PINL. Section \ref{CAFW} concludes the paper and discusses directions for future work.

\section{Preliminaries}\label{PRE}
In this section, we collect the basic order-theoretic, algebraic, and semantic notions used throughout the paper. Since our aim is to develop a positive version of instantial neighbourhood logic, the propositional base is distributive-lattice logic rather than classical propositional logic. Accordingly, we work without classical negation and without implication as a primitive connective, and we formulate derivability in consequence form. The modal semantics is formulated over ordered two-sided neighbourhood structures. This choice is motivated by the positive direction suggested for INL in \cite{bezhanishvili2020duality}, where distributive lattices with instantial operators are indicated as the natural positive analogue of Boolean algebras with instantial operators.\\
For standard background on posets, up-sets, and distributive lattices, we refer to \cite{davey2002introduction}.
\subsection*{Posets and upward closed sets}
A partially ordered set is a pair $(W,\leq)$, where $\leq$ is reflexive, antisymmetric, and transitive. For $X\subseteq W$, we write upward closure of $X$ as 
\[
\uparrow X=\{w\in W:\exists x\in X \text{ such that } x\leq w\}.
\]
A subset $U\subseteq W$ is called an up-set or upward closed if $u\in U$ and $u\leq v$ then $v\in U$. Let $\mathrm{Up}(W,\leq)$ denote the set of all up-sets of $W$.
\subsection*{Distributive lattice background}
A bounded distributive lattice is an algebra 
\[
(D,\vee,\wedge,\bot,\top)
\]
such that $(D,\vee,\wedge)$ is a distributive lattice with least element $\bot$ and greatest element $\top$. The lattice order $\leq$ is defined as 
\[
a\leq b\iff a\wedge b=a\iff a\vee b=b.
\]
Since the logic is positive, the propositional fragment is based on distributive-lattice logic. Thus instead of taking all classical propositional logic axioms, we take as the propositional base the valid laws of bounded distributive lattices. This is exactly the positive analogue of the classical propositional basis, and it matches the Dunn's positive modal logic view point(\cite{dunn1995positive}), where the underlying algebraic semantics is distributive-lattice based rather than Boolean.
\subsection*{Background on Instantial Neighbourhood Logic}
We briefly recall the classical language of Instantial neighbourhood logic, following \cite{van2017instantial}. Let $\mathsf{Prop}$ be an arbitrary but fixed set of propositional letters. The language $\mathcal{L}_{INL}$ of INL is generated recursively by 
\[
\varphi::=\top\mid p\mid \neg\varphi\mid (\varphi\wedge\varphi)\mid \Box(\varphi_1,\ldots,\varphi_n;\varphi), \quad p\in\mathsf{Prop},\ n\in\omega. 
\]
When $n=0$, we write simply $\Box(\psi)$.
As usual, the connectives $\bot$, $\vee$, $\to$, and $\leftrightarrow$ are defined from $\neg$ and $\wedge$ in the standard way. \\

A neighbourhood frame is a pair $\mathfrak F=(W,N)$, where $W$ is a non-empty
set and $N:W\to \mathcal{P}(\mathcal{P}(W))$ is a neighbourhood function, where $\mathcal{P}(W)$ denotes the power set of $W$. A valuation on $W$ is a map $V:\mathsf{Prop}\to \mathcal{P}(W)$. A neighbourhood
model is a triple $\mathfrak M=(W,N,V)$.\\
The truth of INL-formulas is defined inductively in the usual way for the
propositional connectives. For the instantial modal formula, 
\[
\mathfrak M,w \models \Box(\psi_1,\ldots,\psi_n;\varphi)
\]
iff there exists $S\in N(w)$ such that every point of $S$ satisfies
$\varphi$, and for each $i\leq n$ there exists some $s_i\in S$ such that
\[
\mathfrak M,s_i \models \psi_i.
\]
Equivalently, if $\llbracket \chi \rrbracket^{\mathfrak M}$ denotes the truth
set of a formula $\chi$, then
\[
\mathfrak M,w \models \Box(\psi_1,\ldots,\psi_n;\varphi)
\]
iff there exists $S\in N(w)$ such that
\[
S \subseteq \llbracket \varphi \rrbracket^{\mathfrak M}
\quad\text{and}\quad
S \cap \llbracket \psi_i \rrbracket^{\mathfrak M}\neq \varnothing
\ \text{for each } i\leq n.
\]

Thus, the INL-modality expresses that there is a neighbourhood of the current
world whose elements all satisfy the scope formula $\varphi$, while each instance
formula $\psi_i$ is realized somewhere inside that neighbourhood.

Instantial neighbourhood logic \cite{van2017instantial} was introduced to enrich neighbourhood semantics by allowing formulas to express not only universal information about a neighbourhood, but also existential information about the different kinds of worlds occurring inside it. 
The INL study \cite{van2017instantial} emphasizes that ordinary neighbourhood semantics naturally supports a quantifier pattern of the form $\exists\forall$, while the instantial language refines this by tracking which kinds of worlds occur in the neighbourhood. It also explains that the completeness proof for classical INL proceeds by normal forms and matching canonical models.\\
This motivates the positive two-sided setting developed in the next section.

\section{Positive Instantial Neighbourhood Logic: Syntax and Semantics}\label{PNSS}
In this section, we introduce the positive version of instantial neighbourhood logic. Classical INL extends propositional logic by an instantial modality interpreted over neighbourhood structures. In the present setting, however, we work over a positive propositional base and treat the two instantial modalities as independent primitive modal connectives. More precisely, we replace the classical propositional language by the positive propositional language whose algebraic semantics is given by bounded distributive lattices, and we interpret the resulting
modal language over persistent two-sided neighbourhood structures. One
neighbourhood component governs witness behaviour for the $\Box$-type
connective, while the other governs co-witness behaviour for the
$\Diamond$-type connective. This separation is natural in the positive setting, since in the absence of negation the two modal connectives are no longer interdefinable.


\subsection{The language of Positive INL}
Let $\mathsf{Prop}$ be an arbitrary but fixed set of propositional variables. The language $\mathcal{L}_{\mathrm{PINL}}$ of positive instantial neighbourhood logic is defined recursively by
\[
\varphi::=\top\mid\bot \mid p\mid (\varphi\wedge\psi)\mid (\varphi\vee\psi)\mid \Box(\varphi_1,\ldots,\varphi_n;\chi) \mid \Diamond(\varphi_1,\ldots,\varphi_n;\chi),
\]
where $p\in\mathsf{Prop}$ and $n\in\omega$.
When $n=0$, we write simply $\Box(\chi)$ and $\Diamond(\chi)$.\\
Thus Positive INL is obtained from the classical instantial framework by replacing the negation-based propositional language with the positive language of bounded distributive lattices, and by taking two primitive instantial connectives, one of box-type and one of diamond-type, instead of relying on definability by negation.\\
For a modal formula of the form
\[
\Box(\varphi_1,\ldots,\varphi_n;\psi)
\quad\text{or}\quad
\Diamond(\varphi_1,\ldots,\varphi_n;\psi),
\]
where \(\varphi_1,\ldots,\varphi_n,\psi\in\mathcal{L}_{\mathrm{PINL}}\), the formulas \(\varphi_1,\ldots,\varphi_n\) are called the instance formulas, and \(\psi\) is called the scope formula. We use the terms witness and co-witness for the corresponding neighbourhood-set conditions in the semantics.
\subsection*{Witness and co-witness}
We recall the notion of witness and co-witness for neighbourhoods, following \cite{bezhanishvili2020duality}.\\
Let $A_1,....,A_n,B\subseteq W$, and let $U\subseteq W$. We say that $U$ witnesses $(A_1,\ldots,A_n;B)$ if 
\[
U\subseteq B \text{ and } U\cap A_i\neq\varnothing \;\text{ for each } i=1,\ldots, n.
\]
We say that $U$ co-witnesses $(A_1,\ldots,A_n;B)$ if
\[
U\subseteq A_i\; \text{ for some } i\in\{1,2,\ldots,n\} \;\text{ or } U\cap B\neq \varnothing.
\]
For \(n=0\), the witness condition reduces to \(U\subseteq B\), while the co-witness condition reduces to \(U\cap B\neq\varnothing\).
\subsection*{Persistent two-sided neighbourhood frames and models}
\begin{defn}\label{PNF}
A persistent two-sided neighbourhood frame is a structure 
\[
\mathfrak F=(W,\leq,N^\Box,N^\Diamond),
\]
where 
\begin{enumerate}[(i)]
	\item $W$ is a non-empty set and $(W,\leq)$ is a partially ordered set, and 
	\[
	N^\Box,N^\Diamond:W\to\mathcal{P}(\mathrm{Up}(W,\leq))
	\]
	are neighbourhood functions, where $\mathcal{P}(\mathrm{Up}(W,\leq))$ is the power set of $\mathrm{Up}(W,\leq)$. Thus, for each $w\in W$,
	\[
	N^\Box(w)\subseteq \mathrm{Up}(W,\leq) \text{ and } N^\Diamond(w)\subseteq \mathrm{Up}(W,\leq);
	\]
	\item  for every $w\leq v$,
	\[
	N^\Box(w)\subseteq N^\Box(v) \, \text{ and } \, N^\Diamond(v)\subseteq N^\Diamond(w).
	\]
\end{enumerate}
\end{defn}
Thus every $\Box$-neighbourhood and every $\Diamond$-neighbourhood is an upset, the $\Box$-side neighbourhood assignment is monotone upward and the $\Diamond$-side neighbourhood assignment is monotone downward.
\begin{defn}\label{PNM}
	A persistent two-sided neighbourhood model is a tuple 
	\[
	\mathfrak M=(W,\leq,N^\Box,N^\Diamond,V),
	\]
	where $(W,\leq,N^\Box,N^\Diamond)$ is a persistent two-sided neighbourhood frame and 
	\[
	V:\mathsf{Prop}\to \mathrm{Up}(W,\leq)
	\]
	is a valuation assigning to each propositional variable an upset of $W$.
\end{defn}

\subsection{Truth of PINL-formulas}
\begin{defn}[Truth of PINL-formula]\label{TD}
	Let $\mathfrak M=(W,\leq,N^\Box,N^\Diamond,V)$ be a persistent two-sided neighbourhood model, let $w\in W$, and let $\varphi \in\mathcal{L}_{\mathrm{PINL}}$. We recursively define when $\mathfrak M$ satisfies $\varphi$ at the state $w\in W$, written
	\[
	\mathfrak M, w\models\varphi
	\]
	as follows:
	\begin{enumerate}[(i)]
		\item for propositional variables and lattice connectives: 
		\[
		\mathfrak M,w\models p\iff w\in V(p),
		\]
	\[
	\mathfrak M, w\models\top\;\text{ always }, \mathfrak M, w\models\bot\;\text{ never },
	\]
	\[
	\mathfrak M,w\models\varphi\wedge\psi\iff\mathfrak M, w\models\varphi\;\text{ and }\; \mathfrak M, w\models\psi,
	\]
	\[
	\mathfrak M, w\models\varphi\vee\psi\iff\mathfrak M,w\models\varphi\;\text{ or }\;\mathfrak M,w\models\psi.
	\]
\item for the modal formulas:
\[
\mathfrak M, w\models\Box(\varphi_1,\ldots,\varphi_n;\varphi)
\]
\[
\quad\text{ iff}\quad\text{there exists } U\in N^\Box(w)\text{ such that for all } u\in U,\;\; \mathfrak M, u\models\varphi,
\]
\[
\qquad\qquad\text{and for all } i\in\{1,\ldots,n\}
\text{ there exists } u_i\in U \text{ such that }
\mathfrak M,u_i \models \varphi_i,
\]
\\
\[
\mathfrak M,w \models \Diamond(\varphi_1,\ldots,\varphi_n;\varphi)
\]
\[
\quad \text{iff} \quad
\text{for every } S\in N^\Diamond(w),
\text{ either there exists } i\in\{1,\ldots,n\} \text{ such that }
\]
\[
\qquad\qquad\text{for all } s\in S \text{ we have } \mathfrak M,s \models \varphi_i,
\text{ or there exists } s\in S \text{ such that } \mathfrak M,s \models \varphi.
\]
\end{enumerate}
\end{defn}
For each formula $\chi\in\mathcal{L}_{\mathrm{PINL}}$, we write 
\[
\llbracket\chi\rrbracket^{\mathfrak M}=\{w\in W:\mathfrak M, w\models\chi\}
\]
for the truth set of $\chi$ in $\mathfrak M$.\\
Thus 
\[
\mathfrak M, w\models\Box(\varphi_1,\ldots,\varphi_n;\varphi)\iff \exists\; U\in N^\Box(w) \text{ such that } U \text{ witnesses } (\llbracket\varphi_1\rrbracket^{\mathfrak M},\ldots,\llbracket\varphi_n\rrbracket^{\mathfrak M};\llbracket\varphi\rrbracket^{\mathfrak M}).
\]
Further, 
\[
\mathfrak M,w\models\Diamond(\varphi_1,\varphi_2,\ldots,\varphi_n;\varphi)\iff \forall U\in N^\Diamond(w)\text{ co-witnesses } (\llbracket\varphi_1\rrbracket^{\mathfrak M},\ldots,\llbracket\varphi_n\rrbracket^{\mathfrak M};\llbracket\varphi\rrbracket^{\mathfrak M}).
\]

\subsection*{Persistence of truth}
The persistence conditions imposed on the frame ensure that truth is monotone upward in the order.
\begin{lem}\label{POT}
	Let $\mathfrak M=(W,\leq,N^\Box,N^\Diamond,V)$ be a persistent two-sided neighbourhood model. Then for every formula $\varphi$, the truth set $\llbracket\varphi\rrbracket^{\mathfrak M}$ is an upset of $W$. Equivalently, if $w\leq v$ and $\mathfrak M, w\models\varphi$, then $\mathfrak M,v\models\varphi$.
	\begin{proof}
		The proof is by induction on the complexity of $\varphi$. The propositional cases follow from the fact that \(V(p)\) is an up-set for every propositional variable \(p\), and that up-sets are closed under finite intersections and finite unions.\\
		Suppose $w\leq v$ and 
		\[
		\mathfrak M,w\models\Box(\varphi_1,\ldots,\varphi_n;\psi).
		\]
		Then there exists $U\in N^\Box(w)$ such that 
		\[
		U\subseteq\llbracket\psi\rrbracket^{\mathfrak M}\text{ and } U\cap\llbracket\varphi_i\rrbracket^{\mathfrak M}\neq\varnothing \text{ for each } i=1,2,\ldots,n.
		\]
		Since \(w\leq v\), the persistence condition for \(N^\Box\) gives
		\[
		N^\Box(w)\subseteq N^\Box(v).
		\] 
		Hence $U\in N^\Box(v)$. 
		 Therefore, the same $U$ witnesses 
		\[
		(\llbracket\varphi_1\rrbracket^{\mathfrak M},\ldots,
		\llbracket\varphi_n\rrbracket^{\mathfrak M};
		\llbracket\psi\rrbracket^{\mathfrak M}).
		\]
		Hence
		\[
		\mathfrak M,v\models\Box(\varphi_1,\ldots,\varphi_n;\psi).
		\]
		Next suppose $w\leq v$ and 
		\[
		\mathfrak M,w\models\Diamond(\varphi_1,\ldots,\varphi_n;\psi).
		\]
		Let $U\in N^\Diamond(v)$. Since $w\leq v$, the persistence condition for  $N^\Diamond$ gives
		\[
		N^\Diamond(v)\subseteq N^\Diamond(w),
		\]
	Thus $U\in N^\Diamond(w)$. Because $\mathfrak M,w\models\Diamond(\varphi_1,\ldots,\varphi_n;\psi)$, the set $U$ co-witnesses 
		\[
		(\llbracket\varphi_1\rrbracket^{\mathfrak M},\ldots,
	\llbracket\varphi_n\rrbracket^{\mathfrak M};
	\llbracket\psi\rrbracket^{\mathfrak M}).
		\]
		Therefore the same condition is true when evaluated at $v$, and hence 
		\[
		\mathfrak M, v\models\Diamond(\varphi_1,\ldots,\varphi_n;\psi). 
		\]
		This completes the induction.
	\end{proof}
\end{lem}
\section{A Proof System for Positive INL}\label{PSPINL}
In this section, we introduce the deductive system for positive instantial neighbourhood logic. Since the propositional base is positive and implication is not taken as primitive, it is natural to formulate the system in entailment form rather than as a classical Hilbert-style system. Thus the formal objects of the calculus are sequents 
\[
\varphi\vdash\psi,
\]
to be read as saying that $\varphi$ entails $\psi$. We write $\varphi\dashv\vdash\psi$ to mean 
\[
\varphi\vdash\psi\text{ and } \psi\vdash\varphi. 
\]  
A PINL formula $\varphi$ is called a theorem if the sequent $\top\vdash\varphi$ is derivable in the system, denoted by $\vdash\varphi$.\\
The proof system consists of three parts:
\begin{enumerate}
	\item propositional axioms, given by the axioms of bounded distributive lattices;
	\item modal axioms for $\Box$;
	\item modal axioms for $\Diamond$.
\end{enumerate}
The system also contains the structural rules listed below.\\
\subsection*{Structural rules}
We take the following rules of inference.\\

\textbf{(CUT)}
\[
\frac{\varphi\vdash\psi\quad\;\psi\vdash\chi}{\varphi\vdash\chi}
\]
\textbf{(US) Uniform substitution}
\[
\frac{\varphi\vdash\psi}{\sigma(\varphi)\vdash\sigma(\psi)}
\]
where $\sigma$ is a uniform substitution of formulas for propositional variables.\\

\textbf{(RE) Replacement of equivalents}
\[
\frac{\alpha\dashv\vdash\beta\quad\;\varphi\vdash\psi}{\varphi[\beta/\alpha]\vdash\psi[\beta/\alpha]}
\]
where $\varphi[\beta/\alpha]$ denotes the result of replacing occurrences of $\alpha$ in $\varphi$ by $\beta$, and similarly for \(\psi[\beta/\alpha]\).
\subsection*{Propositional axioms}
The propositional base of the system is the positive consequence relation of bounded distributive lattices. Thus the system contains, as axioms schemata, all bounded distributive lattice entailments. In particular, the following are available:
\begin{enumerate}[(i)]
	\item reflexivity:
	\[
	\varphi\vdash\varphi;
	\]
\item commutativity:
\[
\varphi\wedge\psi\dashv\vdash\psi\wedge\varphi,\quad\; \varphi\vee\psi\dashv\vdash\psi\vee\varphi;
\]
\item associativity:
\[
(\varphi\wedge\psi)\wedge\chi\dashv\vdash\varphi\wedge(\psi\wedge\chi),\quad\;
(\varphi\vee\psi)\vee\chi\dashv\vdash\varphi\vee(\psi\vee\chi);
\]
\item idempotence:
\[
\varphi\wedge\varphi\dashv\vdash\varphi,\quad \varphi\vee\varphi\dashv\vdash\varphi;
\]
\item absorption:
\[
\varphi\wedge(\varphi\vee\psi)\dashv\vdash\varphi,\quad\; \varphi\vee(\varphi\wedge\psi)\dashv\vdash\varphi;
\]
\item  distributivity:
\[
\varphi\wedge(\psi\vee\chi)\dashv\vdash(\varphi\wedge\psi)\vee(\varphi\wedge\chi),\quad\;
\varphi\vee(\psi\wedge\chi)\dashv\vdash(\varphi\vee\psi)\wedge(\varphi\vee\chi);
\]
\item bounds:
\[
\begin{aligned}
\bot\vdash\varphi,\quad\;\varphi\vdash\top;\\
\varphi\wedge\top\dashv\vdash\varphi,\quad\;\varphi\vee\bot\dashv\vdash\varphi;\\
\varphi\wedge\bot\dashv\vdash\bot,\quad\;\varphi\vee\top\dashv\vdash\top;
\end{aligned}
\]
\item conjunction elimination:
\[
\varphi\wedge\psi\vdash\varphi,\quad\;\varphi\wedge\psi\vdash\psi;
\]
\item disjunction introduction:
\[
\varphi\vdash\varphi\vee\psi,\quad\; \psi\vdash\varphi\vee\psi.
\]
\end{enumerate}
The following inference rules are admissible in the bounded distributive lattice fragment.
\begin{enumerate}[(i)]
	\item Conjunction introduction$\colon$
	\[
	\frac{\alpha\vdash\beta\quad\;\alpha\vdash\gamma}{\alpha\vdash\beta\wedge\gamma};
	\]
	\item Disjunction elimination$\colon$
	\[
	\frac{\beta\vdash\alpha\quad\;\gamma\vdash\alpha}{\beta\vee\gamma\vdash\alpha};
	\]
	\item  Distributive cut$\colon$
	\[
	\frac{\varphi\wedge\alpha\vdash\psi\quad\;\varphi\vdash\psi\vee\alpha}{\varphi\vdash\psi}.
	\]
\end{enumerate}
\subsection*{Modal axioms for $\Box$}
The modal axioms for $\Box$ are the following.
\begin{enumerate}[(i)]
	\item (\(\Box1\)): for each $1\leq k\leq n$,
	\[
	\Box(\varphi_1,\ldots,\varphi_{k-1},\bot,\varphi_{k+1},\ldots,\varphi_n;\psi)\vdash\bot,
	\]
	\item ($\Box$2) Permutation invariance: for every permutation $\pi$ of $\{1,\ldots,n\}$,
	\[
	\Box(\varphi_1,\ldots,\varphi_n;\psi)\dashv\vdash\Box(\varphi_{\pi(1)},\ldots,\varphi_{\pi(n)};\psi),
	\]
	\item ($\Box$3) Join-preservation in each instance position: for each $1\leq k\leq n$,
	\[
	\begin{aligned}
	\Box(\varphi_1,\ldots,\varphi_{k-1},\alpha\vee\beta,\varphi_{k+1},\ldots,\varphi_n;\psi)\\
	\dashv\vdash\Box(\varphi_1,\ldots,\varphi_{k-1},\alpha,\varphi_{k+1},\ldots,\varphi_n;\psi)\vee\Box(\varphi_1,\ldots,\varphi_{k-1},\beta,\varphi_{k+1},\ldots,\varphi_n;\psi),
\end{aligned}
	\]
	\item ($\Box$4) Meet-subadditivity in the scope coordinate: 
	\[
	\Box(\varphi_1,\ldots,\varphi_n;\psi\wedge\chi)\vdash\Box(\varphi_1,\ldots,\varphi_n;\psi)\wedge\Box(\varphi_1,\ldots,\varphi_n;\chi),
	\]
	\item ($\Box$5) Scope absorption in an instance position:
	\[
	\Box(\varphi_1,\ldots,\varphi_k,\ldots,\varphi_n;\psi)\vdash\Box(\varphi_1,\ldots,\varphi_k\wedge\psi,\ldots,\varphi_n;\psi),\quad 1\leq k\leq n.,
	\]
	\item ($\Box$6) Positive cover axiom: if 
	\[
	\top\dashv\vdash\gamma\vee\delta,
	\]
	then
	\[
	\Box(\varphi_1,\ldots,\varphi_n;\psi)\vdash\Box(\varphi_1,\ldots,\varphi_n,\gamma;\psi)\vee\Box(\varphi_1,\ldots,\varphi_n;\psi\wedge\delta),
	\]
	\item ($\Box$7) Weakening in arity:
	\[
	\Box(\varphi_1,\ldots,\varphi_n,\varphi_{n+1};\psi)\vdash\Box(\varphi_1,\ldots,\varphi_n;\psi),
	\]
	\item ($\Box$8) Duplication: for $n\geq 1$,
	\[
	\Box(\varphi_1,\ldots,\varphi_n;\psi)\vdash\Box(\varphi_1,\ldots,\varphi_n,\varphi_n;\psi).
	\]
\end{enumerate}
\noindent\subsection*{Modal axioms for $\Diamond$}
The modal axioms for $\Diamond$ are adapted to the universal co-witness semantics.
\begin{enumerate}[(i)]
	\item ($\Diamond1$) $\top$ in an instance formula: for each $1\leq k\leq n$,
	\[
	\top\vdash\Diamond(\varphi_1,\ldots,\varphi_{k-1},\top,\varphi_{k+1},\ldots,\varphi_n;\psi)
	\]
	\item ($\Diamond2$) Permutation invariance: for every permutation $\pi$ of $\{1,2,\ldots,n\}$,
	\[
	\Diamond(\varphi_1,\ldots,\varphi_n;\psi)\dashv\vdash\Diamond(\varphi_{\pi(1)},\ldots,\varphi_{\pi(n)};\psi),
	\]
	\item ($\Diamond3$) Meet-preservation in each instance position: for each $1\leq k\leq n$,
	\[
		\Diamond(\varphi_1,\ldots,\varphi_{k-1},\alpha\wedge\beta,\varphi_{k+1},\ldots,\varphi_n;\psi)
		\]
		\[\dashv\vdash\Diamond(\varphi_1,\ldots,\varphi_{k-1},\alpha,\varphi_{k+1},\ldots,\varphi_n;\psi)\wedge\Diamond(\varphi_1,\ldots,\varphi_{k-1},\beta,\varphi_{k+1},\ldots,\varphi_n;\psi),
	\]
	\item ($\Diamond4$) Meet-subadditivity in the scope position: 
	\[
	\Diamond(\varphi_1,\ldots,\varphi_n;\psi\wedge\chi)\vdash\Diamond(\varphi_1,\ldots,\varphi_n;\psi)\wedge\Diamond(\varphi_1,\ldots,\varphi_n;\chi),
	\]
	\item ($\Diamond5$) Absorption in join: 
	\[
	\Diamond(\varphi_1,\ldots,\varphi_{n-1},\varphi_n\vee\psi;\psi)\vdash\Diamond(\varphi_1,\ldots,\varphi_n;\psi),\quad n\geq 1
	\]
	\item ($\Diamond6$) Dual cover axiom: if 
	\[
	\gamma\wedge\delta\dashv\vdash\bot,
	\]
	then
	\[
	\Diamond(\varphi_1,\ldots,\varphi_n,\gamma;\psi)\wedge\Diamond(\varphi_1,\ldots,\varphi_n;\psi\vee\delta)\vdash\Diamond(\varphi_1,\ldots,\varphi_n;\psi),
	\]
	\item ($\Diamond7$) Arity extension:
	\[
	\Diamond(\varphi_1,\ldots,\varphi_n;\psi)\vdash\Diamond(\varphi_1,\ldots,\varphi_n,\varphi_{n+1};\psi),
	\]
	\item ($\Diamond8$) Duplication invariance:
	\[
	\Diamond(\varphi_1,\ldots,\varphi_n,\varphi_n;\psi)\dashv\vdash\Diamond(\varphi_1,\ldots,\varphi_n;\psi),\quad n\geq 1.
	\]
\end{enumerate}
We also assume the following monotonicity rules for the two independent modal operators.
\begin{enumerate}[(i)]
	\item \[
	(\Box \mathrm{Mon})\qquad \frac{\varphi_i\vdash\varphi_i'\;(1\leq i\leq n)\quad\psi\vdash\psi'}{\Box(\varphi_1,\ldots,\varphi_n;\psi)\vdash\Box(\varphi_1',\ldots,\varphi_n';\psi')}.
	\]
	For the \(\Diamond\)-modality we take the analogous monotonicity rule:
	\item \[
		(\Diamond \mathrm{Mon})\qquad \frac{\varphi_i\vdash\varphi_i'\;(1\leq i\leq n)\quad\psi\vdash\psi'}{\Diamond(\varphi_1,\ldots,\varphi_n;\psi)\vdash\Diamond(\varphi_1',\ldots,\varphi_n';\psi')}.
	\]
\end{enumerate}
\begin{lem}\label{ADSD}
	In $\mathrm{PINL}$, the following schemata are derivable.\\
	\begin{enumerate}[(i)]
	\item  \[
		(\mathrm{R\text{-}Mon}_{\Box})\qquad
		\Box(\varphi_1,\ldots,\varphi_n;\psi)\vdash
		\Box(\varphi_1,\ldots,\varphi_n;\psi\vee\chi),
		\]
\item 	\[
		(\mathrm{L\text{-}Mon}_{\Box})\qquad
		\Box(\varphi_1,\ldots,\varphi_i,\ldots,\varphi_n;\psi)\vdash
		\Box(\varphi_1,\ldots,\varphi_i\vee\chi,\ldots,\varphi_n;\psi),
		\]
	\item  \[(\mathrm{R\text{-}Mon}_{\Diamond})\qquad
		\Diamond(\varphi_1,\ldots,\varphi_n;\psi)\vdash
		\Diamond(\varphi_1,\ldots,\varphi_n;\psi\vee\chi),
				\]
		\item 	 \[
				(\mathrm{L\text{-}Mon}_{\Diamond})\qquad
				\Diamond(\varphi_1,\ldots,\varphi_i,\ldots,\varphi_n;\psi)\vdash
				\Diamond(\varphi_1,\ldots,\varphi_i\vee\chi,\ldots,\varphi_n;\psi).
				\]

\end{enumerate}
\begin{proof}
	\begin{enumerate}[(i)]
		\item Since $\psi\vdash\psi\vee\chi$ is derivable in the propositional fragment, an application of  $(\Box\mathrm{Mon})$ in the scope formula yields
		\[
		\Box(\varphi_1,\ldots,\varphi_n;\psi)\vdash
		\Box(\varphi_1,\ldots,\varphi_n;\psi\vee\chi).
		\]
		\item Since \(\varphi_i\vdash \varphi_i\vee\chi\) is derivable in the propositional fragment, an application of \((\Box\mathrm{Mon})\) in the \(i\)-th position of instance formulas yields
		\[
		\Box(\varphi_1,\ldots,\varphi_i,\ldots,\varphi_n;\psi)
		\vdash
		\Box(\varphi_1,\ldots,\varphi_i\vee\chi,\ldots,\varphi_n;\psi).
		\]
		\item  Since $\psi\vdash \psi\vee\chi$ is derivable in the propositional fragment,
		an application of $(\Diamond\mathrm{Mon})$ in the scope position yields
		\[	\Diamond(\varphi_1,\ldots,\varphi_n;\psi)\vdash
		\Diamond(\varphi_1,\ldots,\varphi_n;\psi\vee\chi),
		\]
		\item Similarly, since $\varphi_i\vdash\varphi_i\vee\chi$  is derivable in the propositional fragment, applying $(\Diamond\mathrm{Mon})$ in the \( i \)-th position of instance formulas yields
		 \[
		\Diamond(\varphi_1,\ldots,\varphi_i,\ldots,\varphi_n;\psi)
		\vdash
		\Diamond(\varphi_1,\ldots,\varphi_i\vee\chi,\ldots,\varphi_n;\psi).
		\]
	\end{enumerate}
\end{proof}
\end{lem}
\subsection*{Relation with the original INL-style nomenclature}
It may be useful to point out explicitly how the present positive system relates to the axiom names used in the Boolean INL literature. The original Hilbert system for INL contains schemata usually denoted by $\mathrm{R}\text{-}\mathrm{Mon}, \mathrm{L}\text{-}\mathrm{Mon}, \mathrm{Inst}, \mathrm{Norm}, \mathrm{Case}, \mathrm{Weak}, \text{ and } \mathrm{Dupl}$. In the present positive system, the monotonicity principles \(\mathrm{R}\text{-}\mathrm{Mon}\) and \(\mathrm{L}\text{-}\mathrm{Mon}\) are represented by the corresponding \(\Box\)- and
\(\Diamond\)-monotonicity schemata derived in Lemma \ref{ADSD}. The axiom $(\Box 1)$ is the positive analogue of $\mathrm{Norm}$, $(\Box 5)$ is the positive analogue of $\mathrm{Inst}$, and $(\Box 7)$ and $(\Box 8)$ correspond to $\mathrm{Weak}$ and $\mathrm{Dupl}$, respectively. The axiom $(\Box 6)$ replaces the classical Boolean $\mathrm{Case}$ axiom by its positive cover form. The corresponding $\Diamond$-schemata form the independent dual package.\\
\noindent The system $\mathrm{PINL}$ introduced in this section will be the deductive basis for the subsequent soundness analysis and for the later algebraic treatment by means of $2$-$\mathrm{DLIO}$s. At this stage, we keep the two modal operators fully independent and do not impose any interaction axioms between them.
\section{Soundness of PINL}\label{SCPINL}
In this section, we study the relation between the proof system introduced in Section
\ref{PSPINL} and the persistent two-sided neighbourhood semantics introduced in Section \ref{PNSS}. We prove soundness of PINL with respect to this persistent semantics. A direct canonical completeness proof for the untyped persistent semantics is obstructed by the extensional nature of neighbourhoods. For this reason, the completeness argument will be carried out in the next section using an auxiliary typed persistent semantics. 

\subsection*{Derivability and semantic validity}
We begin by fixing the proof-theoretic and semantic notation that will be used throughout the section.
\begin{defn}
	For any two formulas $\varphi,\psi$ of $\mathcal{L}_{\mathrm{PINL}}$, we write 
	\[
	\varphi\vdash_{\mathrm{PINL}}\psi
	\]
	if the sequent $\varphi\vdash\psi$ is derivable in the proof system for PINL.
\end{defn}
\begin{defn}
	Let $\mathfrak M=(W,\leq,N^\Box,N^\Diamond,V)$ be a persistent two-sided neighbourhood model. We say the sequent 
	\[
	\varphi\vdash\psi
	\]
	is valid in $\mathfrak M$, and write $\mathfrak M\models (\varphi\vdash\psi)$, if 
	\[
	\llbracket\varphi\rrbracket^{\mathfrak M}\subseteq\llbracket\psi\rrbracket^{\mathfrak M}.
	\]
	Equivalently,
	\[
	\mathfrak M\models (\varphi\vdash\psi)
	\]
	iff for every $w\in W$,
	\[
	\mathfrak M,w\models\varphi\Longrightarrow\mathfrak M,w\models\psi.
	\]
\end{defn}
\begin{defn}
	We say that $\varphi\vdash\psi$ is valid in the persistent two-sided neighbourhood semantics, and write 
	\[
	\models(\varphi\vdash\psi),
	\]
	if 
	\[
	\mathfrak M\models (\varphi\vdash\psi)
	\]
	for every persistent two-sided neighbourhood model $\mathfrak M$.
	\end{defn}
\subsection{Soundness for persistent two-sided neighbourhood models}
We now verify that the proof system for $\mathrm{PINL}$ is sound for the persistent semantics.
\begin{thm}
	For all formulas $\varphi,\psi\in\mathcal{L}_{\mathrm{PINL}}$,
	\[
	\varphi\vdash_{\mathrm{PINL}}\psi\quad\Longrightarrow\quad \models (\varphi\vdash\psi).
	\]
	\begin{proof}
		The proof is by induction on the derivation of the sequent
		\[
		\varphi\vdash_{\mathrm{PINL}}\psi.
		\]
		Consider an arbitrary persistent two-sided neighbourhood model $\mathfrak M=(W,\leq,N^\Box,N^\Diamond,V)$. We show that 
		\[
		\mathfrak M\models (\varphi\vdash\psi).
		\]
		Since $\mathfrak M$ was arbitrary, we have $\models (\varphi\vdash\psi)$.\\ 
		Propositional axioms are valid because in every model:
		\[
		\llbracket\alpha\wedge\beta\rrbracket=\llbracket\alpha\rrbracket\cap\llbracket\beta\rrbracket,\; \llbracket\alpha\vee\beta\rrbracket=\llbracket\alpha\rrbracket\cup\llbracket\beta\rrbracket,\quad
		\llbracket\bot\rrbracket=\varnothing,\; \llbracket\top\rrbracket=W.
		\]
	Thus a sequent \(\alpha\vdash\beta\) is valid exactly when
	\[
	\llbracket\alpha\rrbracket^M\subseteq\llbracket\beta\rrbracket^M,
	\]
	and an equivalence \(\alpha\dashv\vdash\beta\) is valid exactly when
	\[
	\llbracket\alpha\rrbracket^M=\llbracket\beta\rrbracket^M.
	\]
	Therefore, all bounded distributive lattice axioms are valid in every persistent two-sided neighbourhood model. For example, the axiom $\bot\vdash\varphi$ is valid because $\varnothing\subseteq\llbracket\varphi\rrbracket^{\mathfrak M}$. Similarly, $\varphi\wedge\psi\vdash\varphi$ is valid because $\llbracket\varphi\rrbracket\cap\llbracket\psi\rrbracket^{\mathfrak M}\subseteq\llbracket\varphi\rrbracket^{\mathfrak M}$.\\
	We now consider the modal axioms. The axioms
	\[
	(\Box2),\ (\Box3),\ (\Box4),\ (\Box7),\ (\Box8)
	\]
	and
	\[
	(\Diamond2),\ (\Diamond3),\ (\Diamond4),\ (\Diamond7),\ (\Diamond8)
	\]
	follow directly from the witness and co-witness clauses. We verify only the non-trivial axioms:
	\[
	(\Box1),\quad(\Box5),\quad(\Box6),\quad(\Diamond1),\quad(\Diamond5),\quad(\Diamond6).
	\]
	For $\Box1$, suppose
	\[
	\mathfrak M,w\models \Box(\varphi_1,\ldots,\varphi_{k-1},\bot,\varphi_{k+1},\ldots,\varphi_n;\psi).
	\]
	Then there is $S\in N^\Box(w)$ such that $S\cap\llbracket\bot\rrbracket^{\mathfrak M}\neq\varnothing$. But $\llbracket\bot\rrbracket^{\mathfrak M}\neq\varnothing$, which is impossible. Hence
	\[
	\llbracket\Box(\varphi_1,\ldots,\varphi_{k-1},\bot,\varphi_{k+1},\ldots,\varphi_n;\psi)\rrbracket^M=\varnothing,
	\]
	and therefore \((\Box1)\) is valid.\\
	For \(\Box5\), suppose
	\[
	\mathfrak M,w\models \Box(\varphi_1,\ldots,\varphi_n;\psi).
	\]
	Then there exists \(S\in N^\Box(w)\) such that
	\[
	S\subseteq \llbracket\psi\rrbracket^M
	\]
	and
	\[
	S\cap\llbracket\varphi_i\rrbracket^{\mathfrak M}\neq\varnothing
	\quad\text{for each }i=1,\ldots,n.
	\]
	In particular, choose
	\[
	x\in S\cap\llbracket\varphi_n\rrbracket^{\mathfrak M}.
	\]
	Since $	S\subseteq \llbracket\psi\rrbracket^M$, we also have
	\[
	x\in\llbracket\psi\rrbracket^{\mathfrak M}.
	\]
	Thus
	\[
	x\in \llbracket\varphi_n\wedge\psi\rrbracket^{\mathfrak M}.
	\]
	Therefore the same \( S \) witnesses $(\llbracket\varphi_1\rrbracket^{\mathfrak M},\ldots,\llbracket\varphi_n\wedge\psi\rrbracket^{\mathfrak M};\llbracket\psi\rrbracket^{\mathfrak M})$. Thus
	\[
	\mathfrak M,w\models \Box(\varphi_1,\ldots,\varphi_{n-1},\varphi_n\wedge\psi;\psi).
	\]
	Hence \((\Box5)\) is valid.\\
	For \(\Box6\), assume 
	\[
	\top\dashv\vdash \gamma\vee\delta.
	\]
	By the soundness of the propositional part, for every persistent two-sided neighbourhood model \(\mathfrak M\), 
	\[
	\llbracket\top\rrbracket^M=\llbracket\gamma\vee\delta\rrbracket^{\mathfrak M}.
	\]
	Hence
	\[
	W=\llbracket\gamma\rrbracket^M\cup\llbracket\delta\rrbracket^{\mathfrak M}.
	\]
	Now suppose
	\[
	\mathfrak M,w\models \Box(\varphi_1,\ldots,\varphi_n;\psi).
	\]
	Then there exists \(S\in N^\Box(w)\) such that
	\[
	S\subseteq \llbracket\psi\rrbracket^{\mathfrak M}
	\]
	and
	\[
	S\cap\llbracket\varphi_i\rrbracket^{\mathfrak M}\neq\varnothing
	\quad\text{for each }i=1,\ldots,n.
	\]
	Now if 
	\[
	S\cap\llbracket\gamma\rrbracket^{\mathfrak M}\neq\varnothing,
	\]
	then the same \( S\) witnesses $(\llbracket\varphi_1\rrbracket^{\mathfrak M},\ldots,\llbracket\varphi_n\rrbracket^{\mathfrak M},\llbracket\gamma\rrbracket^{\mathfrak{M}};\llbracket\psi\rrbracket^{\mathfrak M})$. Hence
	\[
	\mathfrak M,w\models \Box(\varphi_1,\ldots,\varphi_n,\gamma;\psi).
	\]
	If $S\cap\llbracket\gamma\rrbracket^{\mathfrak M}=\varnothing$, then for every $x\in S$, we have $x\notin\llbracket\gamma\rrbracket^{\mathfrak M}$. Since $x\in W=\llbracket\gamma\rrbracket^{\mathfrak M}\cup\llbracket\delta\rrbracket^{\mathfrak M}$, we have $x\in\llbracket\delta\rrbracket^{\mathfrak M}$. Consequently, $S\subseteq\llbracket\delta\rrbracket^{\mathfrak M}$. Hence
	\[
	S\subseteq\llbracket\psi\rrbracket^{\mathfrak M}\cap\llbracket\delta\rrbracket^{\mathfrak M}=\llbracket\psi\wedge\delta\rrbracket^{\mathfrak M}.
	\]  
	So the same $S$ witnesses $(\llbracket\varphi_1\rrbracket^{\mathfrak M},\ldots,\llbracket\varphi_n\rrbracket^{\mathfrak M};\llbracket\psi\wedge\delta\rrbracket^{\mathfrak M})$. Therefore, $w\models\Box(\varphi_1,\ldots,\varphi_n;\psi\wedge\delta)$.\\
	Thus, in either case,
	\[
	\mathfrak M,w\models
	\Box(\varphi_1,\ldots,\varphi_n,\gamma;\psi)
	\vee
	\Box(\varphi_1,\ldots,\varphi_n;\psi\wedge\delta).
	\]
	Hence \((\Box6)\) is valid.\\
	We now verify the \(\Diamond\)-axioms. \\
	For \(\Diamond1\), suppose that the \(k\)-th instance formula is \(\top\). Let $A\in N^\Diamond(w)$. 
	Since
	\[
	\llbracket\top\rrbracket^M=W
	\]
	and \(A\subseteq W\), we have
	\[
	A\subseteq \llbracket\top\rrbracket^{\mathfrak M}.
	\]
	Thus \(A\) co-witnesses $(\llbracket\varphi_1\rrbracket^{\mathfrak M},\ldots,\llbracket\varphi_{k-1},\llbracket\top\rrbracket^{\mathfrak M},\llbracket\varphi_{k+1}\rrbracket^{\mathfrak M},\ldots,\llbracket\varphi_{n}\rrbracket^{\mathfrak M};\llbracket\psi\rrbracket^{\mathfrak M})$. Hence
	\[
	\mathfrak M, w\models
	\Diamond(\varphi_1,\ldots,\varphi_{k-1},\top,\varphi_{k+1},\ldots,\varphi_n;\psi).
	\]
	Therefore, \((\Diamond1)\) is valid.\\
	For \((\Diamond5)\), suppose
	\[
	\mathfrak M,w\models \Diamond(\varphi_1,\ldots,\varphi_{n-1},\varphi_n\vee\psi;\psi).
	\]
	Let \(T\in N^\Diamond(w)\). By the semantic clause for \(\Diamond\), either
	\[
	T\subseteq\llbracket\varphi_i\rrbracket^{\mathfrak M}
	\quad\text{for some }i< n,
	\]
	or
	\[
	T\subseteq\llbracket\varphi_n\vee\psi\rrbracket^\{\mathfrak M\},
	\]
	or
	\[
	T\cap\llbracket\psi\rrbracket^{\mathfrak M}\neq\varnothing.
	\]
If either $T\subseteq\llbracket\varphi_i\rrbracket^{\mathfrak M}
	\quad\text{for some }i< n$ or $T\cap\llbracket\psi\rrbracket^{\mathfrak M}\neq\varnothing$, then the same \(T\) co-witnesses $(\llbracket\varphi_1\rrbracket^{\mathfrak M},\ldots,\llbracket\varphi_n\rrbracket^{\mathfrak M};\llbracket\psi\rrbracket^{\mathfrak M})$. 
	Now consider the case: 
	\[
	T\subseteq\llbracket\varphi_n\vee\psi\rrbracket^\{\mathfrak M\}.
	\]
	If $T\cap\llbracket\psi\rrbracket^{\mathfrak M}\neq\varnothing$, then the required condition holds. Otherwise, 
	\[
		T\cap\llbracket\psi\rrbracket^{\mathfrak M}=\varnothing.
	\]
	Since 
	\[
	\llbracket\varphi_n\vee\psi\rrbracket^\{\mathfrak M\}=\llbracket\varphi_n\rrbracket^{\mathfrak M}\cup\llbracket\psi\rrbracket^{\mathfrak M},
	\]
	it follows s that $T\subseteq \llbracket\varphi_n\rrbracket^{\mathfrak M}$. Then again the required condition holds. Therefore every \(T\in N^\Diamond(w)\) co-witnesses 
	\[
	(\llbracket\varphi_1\rrbracket^{\mathfrak M},\ldots,\llbracket\varphi_n\rrbracket^{\mathfrak M};\llbracket\psi\rrbracket^{\mathfrak M}).
	\]
	Hence 
	\[
	\mathfrak M,w\models\Diamond(\varphi_1,\ldots,\varphi_n;\psi). 
	\]
	So \(\Diamond5\) is valid.\\
	For \((\Diamond6)\), assume
	\[
	\gamma\wedge\delta\dashv\vdash\bot.
	\]
	Then
	\[
	\llbracket\gamma\rrbracket^{\mathfrak M}\cap\llbracket\delta\rrbracket^{\mathfrak M}=\varnothing.
	\]
	Suppose
	\[
	\mathfrak M,w\models \Diamond(\varphi_1,\ldots,\varphi_n,\gamma;\psi)
	\wedge
	\Diamond(\varphi_1,\ldots,\varphi_n;\psi\vee\delta).
	\]
	Then we have both 
	\[
	w\models\Diamond(\varphi_1,\ldots,\varphi_n,\gamma;\psi),
	\]
	and 
	\[
	w\models \Diamond(\varphi_1,\ldots,\varphi_n;\psi\vee\delta).
	\]
	Let \(T\in N^\Diamond(w)\). From
	\[
	\mathfrak M,w\models \Diamond(\varphi_1,\ldots,\varphi_n;\psi\vee\delta),
	\]
	we have either
	\[
	T\subseteq\llbracket\varphi_i\rrbracket^{\mathfrak M}
	\quad\text{for some }i=1,\ldots,n,
	\]
	or
	\[
	T\cap\llbracket\psi\vee\delta\rrbracket^{\mathfrak M}\neq\varnothing.
	\]
	In the first case, the required co-witness condition is satisfied. In the second case,
	either
	\[
	T\cap\llbracket\psi\rrbracket^{\mathfrak M}\neq\varnothing,
	\]
	and again we are done,\\
	or
	\[
	T\cap\llbracket\delta\rrbracket^{\mathfrak M}\neq\varnothing\quad\text{and}\quad
	T\cap\llbracket\psi\rrbracket^{\mathfrak M}=\varnothing.
	\]
	In this last case, using 
	\[
	\mathfrak M,w\models\Diamond(\varphi_1,\ldots,\varphi_n,\gamma;\psi),
	\]
	we get either
	\[
	T\subseteq\llbracket\varphi_i\rrbracket^{\mathfrak M}
	\quad\text{for some }i=1,\ldots,n,
	\]
	or 
	\[
	T\subseteq\llbracket\gamma\rrbracket^{\mathfrak M},
	\]
	or
	\[
	T\cap\llbracket\psi\rrbracket^{\mathfrak M}\neq\varnothing.
	\]
	The last alternative is impossible by the assumption $T\cap\llbracket\psi\rrbracket^{\mathfrak M}=\varnothing$. If  
	\[
	T\subseteq\llbracket\gamma\rrbracket^{\mathfrak M},
	\]
	then together with 
	\[
	T\cap\llbracket\delta\rrbracket^{\mathfrak M}\neq\varnothing
	\]
	we obtain 
	\[
	\llbracket\gamma\rrbracket^{\mathfrak M}\cap\llbracket\delta\rrbracket^{\mathfrak M}\neq\varnothing,
	\]
	contradicting 
	\[
	\llbracket\gamma\rrbracket^{\mathfrak M}\cap\llbracket\delta\rrbracket^{\mathfrak M}=\varnothing.
	\]
	Hence the only possible remaining case is 
	\[
	T\subseteq\llbracket\varphi_i\rrbracket^{\mathfrak M}
	\quad\text{for some }i=1,\ldots,n.
	\]
	Thus \(T\) co-witnesses
	\[
	(\llbracket\varphi_1\rrbracket^{\mathfrak M},\ldots,\llbracket\varphi_n\rrbracket^{\mathfrak M};\llbracket\psi\rrbracket^{\mathfrak M}).
	\]
	Since \(T\in N^\Diamond(w)\) was arbitrary, we have
	\[
	\mathfrak M,w\models \Diamond(\varphi_1,\ldots,\varphi_n;\psi).
	\]
	Therefore \((\Diamond6)\) is valid.\\
	It remains to verify that the inference rules preserve validity. $\mathrm{CUT}$ preserves validity by transitivity of set inclusion.\\
	For the rule \((\Box\mathrm{Mon})\), assume that 
	\[
	\varphi_i\vdash\varphi_i'
	\quad (1\leq i\leq n),
	\qquad
	\psi\vdash\psi'
	\]
	are valid. Then 
	\[
	\llbracket\varphi_i\rrbracket^{\mathfrak M}\subseteq\llbracket\varphi_i'\rrbracket^{\mathfrak M}
	\quad (1\leq i\leq n),
	\qquad
	\llbracket\psi\rrbracket^{\mathfrak M}\subseteq\llbracket\psi'\rrbracket^{\mathfrak M}.
	\]
	Suppose
	\[
	\mathfrak{M},w\models \Box(\varphi_1,\ldots,\varphi_n;\psi).
	\]
	Then there exists \(S\in N^\Box(w)\) such that 
	\[
	S\subseteq\llbracket\psi\rrbracket^{\mathfrak M}
	\]
	and 
	\[
	S\cap\llbracket\varphi_i\rrbracket^{\mathfrak M}\neq\varnothing
	\quad\text{for each }i.
	\]
	By the above inclusions, 
	\[
	S\subseteq\llbracket\psi'\rrbracket^{\mathfrak M}
	\]
	and 
	\[
	S\cap\llbracket\varphi_i'\rrbracket^{\mathfrak M}\neq\varnothing
	\quad\text{for each }i.
	\]
	Therefore, we have
	\[
	\mathfrak{M},w\models \Box(\varphi_1',\ldots,\varphi_n';\psi').
	\]
	Thus \((\Box\mathrm{Mon})\) preserves validity.\\
	For the rule \((\Diamond\mathrm{Mon})\), assume again that
	\[
	\varphi_i\vdash\varphi_i'
	\quad (1\leq i\leq n),
	\qquad
	\psi\vdash\psi'
	\]
	are valid. Then
	\[
	\llbracket\varphi_i\rrbracket^{\mathfrak M}\subseteq\llbracket\varphi_i'\rrbracket^{\mathfrak M}
	\quad (1\leq i\leq n),
	\qquad
	\llbracket\psi\rrbracket^{\mathfrak M}\subseteq\llbracket\psi'\rrbracket^{\mathfrak M}.
	\]
	Suppose
	\[
	\mathfrak{M},w\models \Diamond(\varphi_1,\ldots,\varphi_n;\psi).
	\]
	Let \(T\in N^\Diamond(w)\). Then either
	\[
	T\subseteq\llbracket\varphi_i\rrbracket^{\mathfrak M}
	\quad\text{for some }i,
	\]
	or
	\[
	T\cap\llbracket\psi\rrbracket^{\mathfrak M}\neq\varnothing.
	\]
	In the first case,
	\[
	T\subseteq\llbracket\varphi_i'\rrbracket^{\mathfrak M}.
	\]
	In the second case,
	\[
	T\cap\llbracket\psi'\rrbracket^{\mathfrak M}\neq\varnothing.
	\]
	Hence \(T\) co-witnesses
	\[
	(\llbracket\varphi_1'\rrbracket^{M},\ldots,\llbracket\varphi_n'\rrbracket^{\mathfrak M};\llbracket\psi'\rrbracket^{\mathfrak M}).
	\]
	Since \(T\in N^\Diamond(w)\) was arbitrary, we get
	\[
	\mathfrak{M},w\models \Diamond(\varphi_1',\ldots,\varphi_n';\psi').
	\]
	Thus \((\Diamond\mathrm{Mon})\) preserves validity.\\
	For uniform substitution, suppose that the premise sequent
	\[
	\alpha\vdash\beta
	\]
	is valid in every persistent two-sided neighbourhood model. Let \(\sigma:\mathrm{Prop}\to\mathcal{L}_{\mathrm{PINL}}\) be a uniform substitution. We show that
	\[
	\sigma(\alpha)\vdash\sigma(\beta)
	\]
	is valid. Let
	\[
	\mathfrak M=(W,\leq,N^\Box,N^\Diamond,V)
	\]
	be any persistent two-sided neighbourhood model. Define a new valuation \(V^\sigma:\mathrm{Prop}\to\mathrm{UP}(W,\leq)\) by
	\[
	V^\sigma(p)=\llbracket\sigma(p)\rrbracket^{\mathfrak M}
	\]
	for each propositional variable \(p\). By Lemma \ref{POT}, truth is persistent in \(\mathfrak M\). Hence, for every formula \(\theta\), the truth set
	\[
	\llbracket\theta\rrbracket^{\mathfrak M}
	\]
	is an up-set of \((W,\leq)\). In particular,
	\[
	\llbracket\sigma(p)\rrbracket^{\mathfrak M}\in \operatorname{Up}(W,\leq).
	\]
	Therefore, $V^\sigma$ is a well-defined valuation. Consequently, 
	\[
	\mathfrak{M}^\sigma=(W,\leq,N^\Box,N^\Diamond,V^\sigma).
	\]
	is again a persistent two-sided neighbourhood model. We claim that, for every formula \(\theta\in\mathcal{L}_{\mathrm{PINL}}\) and every \(w\in W\), 
	\[
	\mathfrak M,w\models\sigma(\theta)
	\quad\Longleftrightarrow\quad
	\mathfrak M^\sigma,w\models\theta.
	\]
	The claim is proved by induction on the structure of \(\theta\). The atomic case follows from the definition of \(V^\sigma\). The cases for \(\top,\bot,\wedge,\vee\) are immediate. For the modal cases, the induction hypothesis gives equality of the truth sets of the instance formulas and of the scope formula in the two models; since \(\mathfrak M\) and \(\mathfrak M^\sigma\) have the same \(N^\Box\) and \(N^\Diamond\), the witness and co-witness conditions are preserved. This proves the claim.\\
	Using the claim, we obtain
	\[
	\llbracket\sigma(\alpha)\rrbracket^{\mathfrak M}
	=
	\llbracket\alpha\rrbracket^{\mathfrak M^\sigma},\qquad 	\llbracket\sigma(\beta)\rrbracket^{\mathfrak M}
	=
	\llbracket\beta\rrbracket^{\mathfrak M^\sigma}.
	\]
	Since $(\alpha\vdash\beta)$ is valid in every persistent two-sided neighbourhood model, it is valid in particular in \(\mathfrak M^\sigma\), we have
	\[
	\llbracket\alpha\rrbracket_{\mathfrak{M}^\sigma}\subseteq\llbracket\beta\rrbracket_{\mathfrak{M}^\sigma}.
	\]
	Thus
	\[
	\llbracket\sigma(\alpha)\rrbracket^{\mathfrak M}
	\subseteq
	\llbracket\sigma(\beta)\rrbracket^{\mathfrak M}.
	\]
	Therefore, \(\mathfrak M\models \sigma(\alpha)\vdash\sigma(\beta)\). Since \(\mathfrak M\) was arbitrary, \(\sigma(\alpha)\vdash\sigma(\beta)\) is valid in every persistent two-sided neighbourhood model. Hence uniform substitution preserves validity.\\
	For replacement of equivalents, suppose that the premises
	\[
	\alpha\dashv\vdash\beta
	\qquad\text{and}\qquad
	\varphi\vdash\psi
	\]
	are valid in every persistent two-sided neighbourhood model. So in particular, for the model $\mathfrak M$, we have 
	\[
	\llbracket\alpha\rrbracket^{\mathfrak M}=\llbracket\beta\rrbracket^{\mathfrak M}.\tag{*}
	\]
Also, since $\models(\varphi\vdash\psi)$, we have
	\[
	\llbracket\varphi\rrbracket^{\mathfrak M}
	\subseteq
	\llbracket\psi\rrbracket^{\mathfrak M}.
	\tag{**}
	\]
	We claim that if $\llbracket\alpha\rrbracket^{\mathfrak M}=\llbracket\beta\rrbracket^{\mathfrak M}$ then for every formula \(\theta\in\mathcal{L}_{\mathrm{PINL}}\)
	\[
	\llbracket\theta[\beta/\alpha]\rrbracket_{\mathfrak M}=\llbracket\theta\rrbracket_{\mathfrak M}.
	\]
The claim is proved by induction on the structure of \(\theta\). If \(\theta\) itself is the occurrence of \(\alpha\) being replaced, then
\[
\theta=\alpha
\quad\text{and}\quad
\theta[\beta/\alpha]=\beta.
\]
Hence the claim follows from \((*)\). If \(\theta\) is a propositional variable different from \(\alpha\), or if \(\theta=\top\) or \(\theta=\bot\), then \(\theta[\beta/\alpha]=\theta\), so the claim is immediate.
For compound formulas, the result follows from the induction hypothesis. In the cases of \(\wedge\) and \(\vee\), this is immediate from their interpretation as intersection and union. In the modal cases, the induction hypothesis says that the truth sets of the replaced instance formulas and of the replaced scope formula remain unchanged. Hence the witness condition for \(\Box\) and the co-witness condition for \(\Diamond\) remain unchanged. Therefore the whole formula has the same truth set after replacement.\\
Using the claim, \((**)\) becomes 
\[
\llbracket\varphi[\beta/\alpha]\rrbracket^{\mathfrak M}
\subseteq
\llbracket\psi[\beta/\alpha]\rrbracket^{\mathfrak M}.
\]
Thus
\[
\mathfrak M\models(\varphi[\beta/\alpha]\vdash\psi[\beta/\alpha]).
\]
Since \(\mathfrak M\) was arbitrary,
\[
\models(\varphi[\beta/\alpha]\vdash\psi[\beta/\alpha]).
\]
Hence the rule \((\mathrm{RE})\) preserves validity.\\
Thus we have shown that 
\begin{itemize}
	\item every axiom instance of \(\mathrm{PINL}\) is valid in every persistent two-sided neighbourhood model; and
	\item each inference rule, namely \((\mathrm{US})\), \((\mathrm{CUT})\), \((\mathrm{RE})\), \((\Box\mathrm{Mon})\), and \((\Diamond\mathrm{Mon})\), preserves validity.
\end{itemize}
Therefore, by induction on derivations, every sequent provable in \(\mathrm{PINL}\) is valid in every persistent two-sided neighbourhood model. Hence
\[
\varphi\vdash_{\mathrm{PINL}}\psi
\quad\Longrightarrow\quad
\models(\varphi\vdash\psi).
\]
This completes the proof of soundness.

			\end{proof}
\end{thm}
\begin{rem}
A direct canonical completeness proof for the untyped persistent two-sided neighbourhood semantics is technically difficult, because neighbourhoods in such models are unlabelled sets. Hence the same neighbourhood may witness, or co-witness, different modal formulas. This causes a difficulty in proving the truth lemma, especially in the direction from truth to membership. For this reason, the next section introduces typed persistent two-sided neighbourhood models, where neighbourhoods are recorded together with the modal formulas they are intended to witness or co-witness.
\end{rem}
\section{Canonical Pair Model and Typed Completeness}\label{CPMTC}
In this section we develop a canonical model construction for positive INL. The canonical strategy is inspired by Dunn’s methodology for positive modal logic \cite{dunn1995positive}. Thus the canonical states are taken to be prime theory pairs, rather than maximally consistent sets. A direct canonical proof for the untyped persistent two-sided neighbourhood semantics is technically obstructed by the fact that ordinary neighbourhoods are unlabelled sets. To make this obstruction clear, we first illustrate it in the \(\Box\)-case, and then introduce typed persistent two-sided neighbourhood models, where neighbourhoods are recorded together with the modal formula for which they are intended to serve as witnesses or co-witnesses.
\begin{exmp}
	We illustrate the obstruction in the \(\Box\)-case. Let $W^0$ be the set of all prime theory pairs. We write an element of \(W^0\) as
	\[
	a=(a_1,a_2),
	\]
	where \(a_1\) is the positive theory component and \(a_2\) is the counter-theory component. Define an order \(\leq^0\) on \(W^0\) by
	\[
	a\leq^0 b
	\quad\Longleftrightarrow\quad
	a_1\subseteq b_1.
	\]
	Equivalently, since every formula belongs to exactly one component of a prime theory pair,
	\[
	a\leq^0 b
\quad\Longleftrightarrow\quad
	b_2\subseteq a_2.
	\]
	Then $(W^0,\leq^0)$ is a poset.
	For each formula \(\theta\in\mathcal{L}_{\mathrm{PINL}}\), put
	\[
	\widehat{\theta}=\{a\in W^0:\theta\in a_1\}.
	\]
	Define a valuation $V^0:\mathrm{Prop}\to\mathrm{UP}(W,\leq^0)$ by
	\[
	V^0(p)=\widehat{p}.
	\]
	This is well-defined, since \(\widehat p\) is an up-set of \((W^0,\leq^0)\). Indeed, if \(a\in\widehat p\) and \(a\leq^0 b\), then \(p\in a_1\subseteq b_1\), so \(b\in\widehat p\). The canonical $\Box$-neighbourhood assignment is defined as 
	\[
		N^{\Box,0}(a)=\{S\in \mathrm{Up}(W^0,\leq^0):\text{ for some } \Box(\theta;\rho)\in a_1\text{ such that } S\subseteq\widehat{\rho}\text{ and } S\cap\widehat{\theta}\neq\varnothing\},
	\]
	and the canonical $\Diamond$-neighbourhood assignment is defined as
	\[
	N^{\Diamond,0}(a)=\{A\in\mathrm{Up}(W^0,\leq^0): \text{ for some } \Diamond(\varphi;\psi)\in a_2, A\not\subseteq\widehat{\varphi}\text{ and } A\cap\widehat{\psi}=\emptyset\}.
	\]
	It is straightforward to check that \(\mathfrak M^0=(W^0,\leq^0,N^{\Box,0},N^{\Diamond,0},V^0)\) is a persistent two-sided neighbourhood model. Indeed, by definition every member of \(N^{\Box,0}(a)\) and \(N^{\Diamond,0}(a)\) is an up-set of \((W^0,\leq^0)\). Moreover, if \(a\leq^0 b\), then \(a_1\subseteq b_1\) and hence \(b_2\subseteq a_2\). To show
	\[
	N^{\Box,0}(a)\subseteq N^{\Box,0}(b),
	\]
	take \(U\in N^{\Box,0}(a)\). By definition, there exists a formula \(\Box(\theta;\rho)\in a_1\) such that
	\[
	U\subseteq\widehat{\rho}
	\quad\text{and}\quad
	U\cap\widehat{\theta}\neq\emptyset.
	\]
	Since \(a_1\subseteq b_1\), we also have
	\[
	\Box(\theta;\rho)\in b_1.
	\]
	Thus the same set \(U\) satisfies the defining condition for membership in \(N^{\Box,0}(b)\). Hence
	\[
	U\in N^{\Box,0}(b).
	\]
	Therefore
	\[
	N^{\Box,0}(a)\subseteq N^{\Box,0}(b).
	\]
	Similarly, the definition of \(N^{\Diamond,0}\) gives
	\[
	N^{\Diamond,0}(b)\subseteq N^{\Diamond,0}(a).
	\]
	Now choose \(a,b,c\in W^0\) such that
	\[
	\Box(p;q)\in a_1
	\qquad\text{and}\qquad
	\Box(r;q)\notin a_1,
	\]
	and
	\[
	p,q\in b_1,\quad r\notin b_1,
	\qquad
	r,q\in c_1,\quad p\notin c_1.
	\]
	Then \(b\) and \(c\) are incomparable with respect to \(\leq^0\). Indeed, \(p\in b_1\) but \(p\notin c_1\), so \(b_1\nsubseteq c_1\), and hence \(b\nleq^0 c\). Similarly, \(r\in c_1\) but \(r\notin b_1\), so \(c_1\nsubseteq b_1\), and hence \(c\nleq^0 b\).
	
	Define
	\[
	S:=\uparrow\{b,c\}
	=
	\{x\in W^0:b\leq^0 x\text{ or }c\leq^0 x\}.
	\]
	Then \(S\) is an up-set of \((W^0,\leq^0)\), namely the smallest up-set containing \(b\) and \(c\).
	
	Since \(q\in b_1\) and \(q\in c_1\), we have
	\[
	b,c\in \widehat q.
	\]
	Also, \(\widehat q\) is an up-set of \((W^0,\leq^0)\). Hence, since \(S=\uparrow{b,c}\) is the smallest up-set containing \(b\) and \(c\), we get
	\[
	S\subseteq \widehat q.
	\]
	Also,
	\[
	S\cap\widehat p\neq\emptyset,
	\]
	because \(b\in S\) and \(p\in b_1\), so \(b\in S\cap\widehat p\). Similarly,
	\[
	S\cap\widehat r\neq\emptyset,
	\]
	because \(c\in S\) and \(r\in c_1\), so \(c\in S\cap\widehat r\).
	
	Therefore \(S\) satisfies the witness conditions for \(\Box(p;q)\), namely
	\[
	S\subseteq \widehat q
	\qquad\text{and}\qquad
	S\cap\widehat p\neq\emptyset.
	\]
	Since \(\Box(p;q)\in a_1\), it follows from the definition of \(N^{\Box,0}(a)\) that
	\[
	S\in N^{\Box,0}(a).
	\]
	Thus \(S\) is admitted as a $\Box$-neighbourhood of \(a\).
	
	However, the same unlabelled set \(S\) also satisfies
	\[
	S\subseteq \widehat q
	\qquad\text{and}\qquad
	S\cap\widehat r\neq\emptyset.
	\]
	Since \(q\) and \(r\) are propositional variables, their semantic truth sets in \(\mathfrak M^0\) are
	\[
	\llbracket q\rrbracket^{\mathfrak M^0}=\widehat q
	\qquad\text{and}\qquad
	\llbracket r\rrbracket^{\mathfrak M^0}=\widehat r.
	\]
	Hence, by the ordinary untyped \(\Box\)-clause,
	\[
	\mathfrak M^0,a\models \Box(r;q).
	\]
	But by assumption,
	\[
	\Box(r;q)\notin a_1.
	\]
	Thus the truth-to-membership direction of the canonical truth lemma may fail for this naive untyped canonical construction.

\end{exmp}

We now introduce typed persistent two-sided neighbourhood models. In the canonical typed model, worlds are prime theory pairs \(a=(a_1,a_2)\), and the truth lemma takes the form
\[
\mathfrak{M}^{\tau,c},a\models_\tau\varphi
\quad\Longleftrightarrow\quad
\varphi\in a_1.
\]
This yields a completeness theorem for Positive INL with respect to typed persistent two-sided neighbourhood models.


\begin{defn}\label{TPNM}
	A typed persistent two-sided neighbourhood model is a structure 
	\[
	\mathfrak{M}^{\tau}=(W,\leq,N_\tau^\Box,N_\tau^\Diamond,V)
	\]
	such that
	\begin{enumerate}[(i)]
		\item $(W,\leq)$ is a non-empty poset;
		\item $V:\mathrm{Prop}\to \mathrm{Up}(W,\leq)$ is a valuation;
		\item for each $w\in W$,
		\[
		\begin{aligned}
		N_\tau^\Box(w)&\subseteq &\mathrm{Form}^\Box\times\mathcal{P}(W),\\
		N_\tau^\Diamond(w)&\subseteq &\mathrm{Form}^\Diamond\times\mathcal{P}(W),
		\end{aligned}
		\]
		where $\mathrm{Form}^\Box$ denotes the set of all formulas of the form 
		\[
		\Box(\varphi_1,\ldots,\varphi_n;\psi),
		\]
		and $\mathrm{Form}^\Diamond$ denotes the set of all formulas of the form 
		\[
		\Diamond(\varphi_1,\ldots,\varphi_n;\psi);
		\]
		\item for any $w,v\in W$ whenever $w\leq v$,
		\[
		N_\tau^\Box(w)\subseteq N_\tau^\Box(v)\text{ and } N_\tau^\Diamond(v)\subseteq N_\tau^\Diamond(w).
		\]
	\end{enumerate}
\end{defn}
\begin{defn}\label{TPMT}
Let $\mathfrak{M}^\tau=(W,\leq,N_\tau^\Box,N_\tau^\Diamond,V)$ be a typed persistent model. The truth relation 
\[
\mathfrak{M}^\tau, w\models_\tau\varphi
\]
is defined inductively in the usual way for $\top,\bot,p,\wedge,\vee$, and for the modal formulas by:
\[
\mathfrak{M}^\tau, w\models_\tau\Box(\psi_1,\ldots,\psi_n;\chi)
\]
iff there exists 
\[
(\Box(\psi_1,\ldots,\psi_n;\chi),S)\in N_\tau^\Box(w)
\]
such that 
\[
S\subseteq\llbracket\chi\rrbracket^{\mathfrak{M}^\tau}\text{ and } S\cap\llbracket\psi_i\rrbracket^{\mathfrak{M}^\tau}\neq\emptyset\text{ for all } i=1,\ldots,n.
\]
Further,
\[
\mathfrak{M}^\tau, w\models_\tau\Diamond(\psi_1,\ldots,\psi_n;\chi)
\]
iff there is no pair
\[
(\Diamond(\psi_1,\ldots,\psi_n;\chi),A)\in N_\tau^\Diamond(w)
\]
such that 
\[
A\cap\llbracket\chi\rrbracket^{\mathfrak{M}^\tau}=\varnothing\text{ and } A\not\subseteq\llbracket\psi_i\rrbracket^{\mathfrak{M}^\tau}\text{ for all } i=1,\ldots,n.
\]
\end{defn}
\begin{defn}\label{TPTD}
	For typed persistent models, we write
	\[
	\mathfrak{M}^\tau\models_\tau (\varphi\vdash\psi)
	\]
	if
	\[
	\llbracket\varphi\rrbracket^{\mathfrak{M}^\tau}\subseteq\llbracket\psi\rrbracket^{\mathfrak{M}^\tau}.
	\]
	We write 
	\[
	\models_\tau(\varphi\vdash\psi)
	\]
	if this is holds in every typed persistent model.
\end{defn}
\begin{prop}[Truth preserving typed expansion of an ordinary model]\label{prop1}
	Every persistent two-sided neighbourhood model canonically induces a typed persistent model preserving truth of all formulas in the language of $\mathrm{PINL}$.
	\begin{proof}
		Let $\mathfrak{M}=(W,\leq,N^\Box,N^\Diamond,V)$ be a persistent two-sided neighbourhood model. Define $\mathfrak M^\tau=(W,\leq,N_\tau^\Box,N_\tau^\Diamond,V)$ by
		\[
		\begin{aligned}
		N_\tau^\Box(w)&=&\{(\alpha,S)\in\mathrm{Form}^\Box\times\mathcal{P}(W):S\in N^\Box(w)\},\\
		N_\tau^\Diamond(w)&=&\{(\delta,A)\in\mathrm{Form}^\Diamond\times\mathcal{P}(W):A\in N^\Diamond(w)\}.
		\end{aligned}
		\]
		Then we show that $\mathfrak M^\tau$ is a typed persistent two-sided neighbourhood model and for every formula $\varphi\in\mathcal{L}_{\mathrm{PINL}}$ and every $w\in W$,
		\[
		\mathfrak M,w\models\varphi\iff\mathfrak M^\tau,w\models_\tau\varphi.
		\]
		Since $N^\Box(w),N^\Diamond(w)\subseteq\mathrm{Up}(W,\leq)\subseteq\mathcal{P}(W)$, the two typed neighbourhood assignments have the required codomains. Furthermore, if $w\leq v$, then persistence of $\mathfrak M$ gives
		\[
		N^\Box(w)\subseteq N^\Box(v),\quad\; N^\Diamond(v)\subseteq N^\Diamond(w).
		\]
		Hence
		\[
		N_\tau^\Box(w)\subseteq N_\tau^\Box(v),\quad\; N_\tau^\Diamond(v)\subseteq N_\tau^\Diamond(w).
		\]
		Thus $\mathfrak M^\tau$ is a typed persistent model.\\
				
		\noindent We now prove the truth equivalence by induction on $\varphi$. \\
		The propositional cases are immediate because $\mathfrak M$ amd $\mathfrak M^\tau$ have the same valuation and interpret \(\top,\bot,\wedge,\vee\) in the same way. Let 
		\[
		\alpha=\Box(\psi_1,\ldots,\psi_n;\chi).
		\]
		Using the induction hypothesis for $\psi_1,\ldots,\psi_n,\chi$, we obtain
		\[
		\begin{aligned}
		\mathfrak M,w\models \alpha&\iff \exists S\in N^\Box(w)\Bigl [S\subseteq\llbracket\chi\rrbracket^{\mathfrak M}
		\text{ and }
		S\cap\llbracket\psi_i\rrbracket^{\mathfrak M}\neq\varnothing
		\text{ for every }i
		\Bigr]
		\\
		&\iff
		\exists (\alpha,S)\in N_{\tau}^{\Box}(w)\,
		\Bigl[
		S\subseteq\llbracket\chi\rrbracket^{\mathfrak M^{\tau}}
		\text{ and }
		S\cap\llbracket\psi_i\rrbracket^{\mathfrak M^{\tau}}
		\neq\varnothing
		\text{ for every }i
		\Bigr]
		\\
		&\iff
		\mathfrak M^{\tau},w\models_{\tau}\alpha.
	\end{aligned}
		\]
		Now let
		\[
		\delta=\Diamond(\varphi_{1},\ldots,\varphi_{n};\psi).
		\]
	The ordinary $\Diamond$-clause is equivalent to the non-existence of a neighbourhood which violates the co-witness condition. Thus, using the induction hypothesis and the definition of \(N_{\tau}^{\Diamond}\), we obtain
		\[
		\begin{aligned}
			\mathfrak M,w\models\delta
			&\Longleftrightarrow
			\neg\exists A\in N^{\Diamond}(w)\,
			\Bigl[
			A\cap\llbracket\psi\rrbracket^{\mathfrak M}=\varnothing
			\text{ and }
			A\nsubseteq\llbracket\varphi_i\rrbracket^{\mathfrak M}
			\text{ for every }i
			\Bigr]
		\\
		&\Longleftrightarrow
		\neg\exists (\delta,A)\in N_{\tau}^{\Diamond}(w)\,
		\Bigl[
		A\cap\llbracket\psi\rrbracket^{\mathfrak M^{\tau}}
		=\varnothing
		\text{ and }
		A\nsubseteq
		\llbracket\varphi_i\rrbracket^{\mathfrak M^{\tau}}
		\text{ for every }i
		\Bigr]
		\\
		&\Longleftrightarrow
		\mathfrak M^{\tau},w\models_{\tau}\delta.
		\end{aligned}
		\]

	\end{proof}
\end{prop}
\begin{rem}
	Proposition \ref{prop1} shows that every ordinary persistent two-sided neighbourhood model \(\mathfrak M\) induces a typed persistent two-sided neighbourhood model \(\mathfrak M^\tau\) preserving truth of all \(\mathrm{PINL}\)-formulas. Consequently,
	\[
	\models_\tau(\varphi\vdash\psi)
	\quad\Longrightarrow\quad
	\models(\varphi\vdash\psi).
	\]
	The converse implication is not established here. 
\end{rem}
In the sequel, typed semantics is used only as an auxiliary canonical semantics for proving the truth lemma and the typed completeness theorem.
\subsection{Canonical Preliminaries}
We now turn to the canonical construction for the typed persistent semantics. Since \(\mathrm{PINL}\) is a positive logic, we adapt Dunn's canonical-pair method for positive logics \cite{dunn1995positive}. Thus the canonical states are taken to be prime theory pairs rather than maximally consistent sets. Accordingly, we recall the notions of theory, counter-theory, disjoint theory pair, and prime theory pair, all taken with respect to the consequence relation \(\vdash_{\mathrm{PINL}}\). These notions will be used in the separation-extension lemma and in the construction of the canonical typed model.

\begin{defn}\label{thpinl}
	A theory is a set $T\subseteq\mathcal{L}_{\mathrm{PINL}}$ such that 
	\begin{enumerate}[(i)]
		\item $\top\in T$;
		\item if $\alpha,\beta\in T$, then $\alpha\wedge\beta\in T$;
		\item if $\alpha\in T$ and $\alpha\vdash_{\mathrm{PINL}}\gamma$, then $\gamma\in T$.
	\end{enumerate}
	\end{defn}
A theory $T$ is proper if $\bot\notin T$.
\begin{defn}\label{cthpinl}
	A counter-theory is a set $F\subseteq\mathcal{L}_{\mathrm{PINL}}$ such that 
	\begin{enumerate}[(i)]
		\item $\bot\in F$;
		\item if $\alpha,\beta\in F$ then $\alpha\vee\beta\in F$;
		\item if $\alpha\vdash_{\mathrm{PINL}}\gamma$ and $\gamma\in F$, then $\alpha\in F$.
	\end{enumerate}
\end{defn}
A counter-theory $F$ is proper if $\top\notin F$. The intuition is that a theory is closed upward under provable consequence, while a counter-theory is closed downward.
\begin{defn}
	A theory pair is a pair 
	\[
	a=(a_1,a_2)
	\]
	such that $a_1$ is a theory and $a_2$ is a counter-theory. \\
	It is called 
	\begin{itemize}
	\item disjoint if $a_1\cap a_2=\varnothing$,
	\item proper if $a_1$ and $a_2$ are both proper.
\end{itemize}
\end{defn}
\begin{defn}
	A proper theory $T$ is called prime if 
	\[
	\varphi\vee\psi\in T\Longrightarrow\varphi\in T\text{ or } \psi\in T.
	\]
	A proper counter-theory $F$ is called prime if 
	\[
	\varphi\wedge\psi\in F\Longrightarrow \varphi\in F\text{ or } \psi\in F.
	\]
	A prime theory pair is a disjoint proper pair
	\[
	a=(a_1,a_2)
	\]
	such that $a_1$ is a prime theory, $a_2$ is a prime counter-theory, and every formula belongs to exactly one of $a_1$ and $a_2$.
\end{defn}
Thus prime theory pairs are the positive counterparts of classical maximal consistent sets.
\begin{defn}
	A set $\Gamma\subseteq\mathcal{L}_{\mathrm{PINL}}$ is consistent if there is no $\gamma_1,\ldots,\gamma_n\in \Gamma$ such that 
	\[
	\gamma_1\wedge\gamma_2\wedge\ldots\wedge\gamma_n\vdash_{\mathrm{PINL}}\bot.
	\]
\end{defn}
\subsection*{Disjointness and prime-pair extension}
We now recall the basic seperation criterion from \cite{dunn1995positive}.
\begin{lem}\label{DP:1}
	A theory pair $a=(a_1,a_2)$ is not disjoint if and only if there exist $\alpha\in a_1$ and $\beta\in a_2$ such that 
	\[
	\alpha\vdash_{\mathrm{PINL}}\beta.
	\]
\end{lem}
For a theory $T$ and a formula $\chi$, let $T+\chi$ denote the theory generated by $T\cup\{\chi\}$. Dually, for a counter-theory $F$, let $F+\chi$ denote the counter-theory generated by $F\cup\{\chi\}$.
\begin{lem}\label{DP:2}
	Let $a=(a_1,a_2)$ be a disjoint theory pair and let $\chi$ be any formula. Then at least one of the pairs 
	\[
	(a_1+\chi,a_2) \text{ or } (a_1,a_2+\chi)
	\]
	is disjoint.
\end{lem}
\begin{lem}\label{DP:3}
	Every disjoint theory pair extends to a prime theory pair.
\end{lem}
\subsection*{Canonical state space and valuation}
Let $W^c=\{a=(a_1,a_2): a \text{ is a prime theory pair} \}$. Define
\[
a\leq^c b\iff a_1\subseteq b_1.
\]
Since every formula belongs to exactly one side of a prime pair, this is equivalent to $b_2\subseteq a_2$.\\
For each formula $\varphi$, define 
\[
\llbracket\varphi\rrbracket^c=\{a\in W^c:\varphi\in a_1\},
\]
and for each variable $p$, define 
\[
V^c(p)=\llbracket p\rrbracket^c.
\]
\begin{lem}\label{L:1}
	For every formula $\varphi\in\mathcal{L}_{\mathrm{PINL}}$, the set $\llbracket\varphi\rrbracket^c$ is upward closed in $(W^c,\leq^c)$.
	\begin{proof}
		If $a=(a_1,a_2)\leq^c b=(b_1,b_2)$ and $a\in \llbracket\varphi\rrbracket^c$, then $\varphi\in a_1\subseteq b_1$. Hence $b\in \llbracket\varphi\rrbracket^c$.
	\end{proof}
\end{lem}
\subsection{Canonical typed $\Box$-neighbourhoods and typed $\Diamond$-neighbourhoods}
\begin{lem}\label{WA}
	Let $a=(a_1,a_2)\in W^c$, and suppose that
	\[
	\alpha=\Box(\psi_1,\ldots,\psi_n;\chi)\in a_1.
	\]  
	Then there exists a finite set 
	\[
	S\subseteq W^c
	\]
	such that 
		\[
	S\subseteq\llbracket\chi\rrbracket^c\text{ and } S\cap\llbracket\psi_i\rrbracket^c\neq\varnothing\text{ for every } i\leq n.
	\]
	\begin{proof}
		We construct, for each $i=1,\ldots,n$, a prime theory pair $u_i\in W^c$ such that 
	\[
	\psi_i\in (u_i)_1\text{ and } \chi\in (u_i)_1.
	\]
	Once this is done, the required finite set will be 
	\[
		S=\{u_1,\ldots,u_n\}.
		\]
		If $n=0$, take $S=\emptyset$. Then $S\subseteq\llbracket\chi\rrbracket^c$, and there are no intersection condition to verify. Hence suppose $n\geq 1$. Let $i\in\{1,\ldots,n\}$. We first show that $\psi_i\wedge\chi$ is consistent i.e., 
		\[
		\psi_i\wedge\chi\not\vdash_{\mathrm{PINL}}\bot.
		\]
		If this is false then 
	\[
	\psi_i\wedge\chi\vdash_{\mathrm{PINL}}\bot.
	\]
	By the axiom $(\Box 2)$ and closure of \(a_1\) under provable consequence, we obtain 
	\[
	\Box(\psi_1,\ldots,\psi_{i-1},\psi_{i+1},\ldots,\psi_n,\psi_i;\chi)\in a_1.
	\]
	Now applying the axiom $(\Box 5)$, we have 
	\[
	\Box(\psi_1\ldots,\psi_{i-1},\psi_{i+1},\ldots,\psi_n,\psi_i\wedge\chi;\chi)\in a_1.
	\]
	Next using the $(\Box\mathrm{Mon})$ and $\psi_i\wedge\chi\vdash_{\mathrm{PINL}}\bot$ together with reflexivity $\chi\vdash_{\mathrm{PINL}}\chi$, we derive that 
	\begin{equation}\label{WA:1}
	\Box(\psi_1\ldots,\psi_{i-1},\psi_{i+1},\ldots,\psi_n,\psi_i\wedge\chi;\chi)\vdash_{\mathrm{PINL}} \Box(\psi_1\ldots,\psi_{i-1},\psi_{i+1},\ldots,\psi_n,\bot;\chi).
\end{equation}
	
	By $(\Box 1)$, 
	
	\begin{equation}\label{WA:2}
		\Box(\psi_1\ldots,\psi_{i-1},\psi_{i+1},\ldots,\psi_n,\bot;\chi)\vdash_{\mathrm{PINL}}\bot.
	\end{equation}
	
	From (\ref{WA:1}) and (\ref{WA:2}), by $\mathrm{Cut}$,
	\[
	\Box(\psi_1\ldots,\psi_{i-1},\psi_{i+1},\ldots,\psi_n,\psi_i\wedge\chi;\chi)\vdash_{\mathrm{PINL}}\bot.
	\]
	Since $\Box(\psi_1\ldots,\psi_{i-1},\psi_{i+1},\ldots,\psi_n,\psi_i\wedge\chi;\chi)\in a_1$ and $a_1$ is a theory, closure under provable consequence yields
	\[
	\bot\in a_1.
	\]
	Since $a_1$ is a proper theory, this is impossible. Hence $\psi_i\wedge\chi$ is consistent.\\
	
	\noindent Define $T_i^0=Th(\{\psi_i,\chi\})$, the theory generated by the set $\{\psi_i,\chi\}$. We claim that $T_i^0$ is proper. If $T_i^0$ is not proper, then $\bot\in T_i^0$. Then we have 
	\[
	\psi_i\wedge\chi\vdash_{\mathrm{PINL}}\bot.
	\]
	But this contradicts the fact that $\psi_i\wedge\chi$ is consistent. Hence $T_i^0$ is a proper theory.\\
	\noindent Define $F^0=CTh(\{\bot\})$, the counter theory generated by $\bot$. Explicitly, $F^0=\{\theta\in\mathcal{L}_{\mathrm{PINL}}:\theta\vdash_{\mathrm{PINL}}\bot\}$. We claim that $F^0$ is proper. If $\top\in F^0$, then by definition 
	\[
	\top\vdash_{\mathrm{PINL}}\bot,
	\]
	which would make the whole positive consequence relation trivial. Since our system is non-trivial, this is impossible. Hence $F^0$ is proper.\\
	
	\noindent We now show that 
	\[
	(T_i^0,F^0) \text{ is a disjoint theory pair }.
	\]
	Otherwise, by Lemma \ref{DP:1} there exist formulas 
	\[
	\alpha\in T_i^0, \;\beta\in F^0
	\]
	such that $\alpha\vdash_{\mathrm{PINL}}\beta$. Since $\beta\in F^0$, we also have 
	\[
	\beta\vdash_{\mathrm{PINL}}\bot.
	\]
	Hence, by $\mathrm{CUT}$
	\[
	\alpha\vdash_{\mathrm{PINL}}\bot.
	\]
	As $\alpha\in T_i^0$ and $T_i^0$ is a theory, closure under consequence yields
	\[
	\bot\in T_i^0, 
	\]
	contradicting the properness of $T_i^0$. Therefore $(T_i^0,F^0)$ is disjoint.\\
		\noindent By Lemma \ref{DP:3}, there exists a prime theory pair
	\[
	u_i=((u_i)_1,(u_i)_2)\in W^c
	\]
	such that 
	\[
	T_i^0\subseteq (u_i)_1\;\text{ and }\; F^0\subseteq (u_i)_2.
	\]
	Since $\psi_i,\chi\in T_i^0$, we have 
		\begin{equation}\label{CDE:3}
	\psi_i\in (u_i)_1\;\text{ and }\; \chi\in (u_i)_1.
\end{equation}
	Since $i$ was arbitrary, we have constructed $u_1,\ldots,u_n\in W^c$ such that for each $i$,
	\[
	\psi_i,\chi\in(u_i)_1.
	\]
	Now define 
	\[
	S=\{u_1,\ldots,u_n\}\subseteq W^c.
	\]
	Let $u\in S$. Then $u=u_i$ for some $i$, and by (\ref{CDE:3}) we have $\chi\in (u_i)_1$. Hence 
	\[
	u\in\llbracket\chi\rrbracket^c.
	\]
	Therefore, $S\subseteq\llbracket\chi\rrbracket^c$. Also for each $i$, we have $\psi_i\in(u_i)_1$ by (\ref{CDE:3}). Hence
	\[
	u_i\in\llbracket\psi_i\rrbracket^c.
	\]
	Since $u_i\in S$, it follows that 
	\[
	u_i\in S\cap\llbracket\psi_i\rrbracket^c.
	\]
	Therefore
	\[
	S\cap\llbracket\psi_i\rrbracket^c\neq\varnothing \;\text{ for every }\; i=1,\ldots, n.
	\]
		\end{proof}
\end{lem}
\noindent Lemma \ref{WA} shows that, whenever a $\Box$-formula belongs to the positive component of a canonical state, there exists at least one finite set satisfying the required witness conditions. We therefore define the canonical typed $\Box$-neighbourhood assignment by collecting all such typed witness pairs.
\begin{defn}\label{WAL:1}
	For each canonical state \(a=(a_{1},a_{2})\in W^{c}\), define
	\[
	\begin{aligned}
		N_{\tau}^{\Box,c}(a)
		:=
		\bigl\{\, &
		(\alpha,S)\in
		\mathrm{Form}^{\Box}\times\mathcal{P}_{\mathrm{fin}}(W^{c})
		:
		\\
		&
		\alpha=\Box(\psi_{1},\ldots,\psi_{n};\chi)\in a_{1},
		\\
		&
		S\subseteq \llbracket \chi \rrbracket^{c},
		\\
		&
		S\cap \llbracket \psi_{i} \rrbracket^{c}\neq\varnothing
		\quad\text{for every } i=1,\ldots,n
		\,\bigr\},
	\end{aligned}
	\]
	where $\mathcal{P}_{\mathrm{fin}}(W^{c})$ denotes the set of all finite subsets of $W^c$.
	\end{defn}

\noindent\subsection*{Canonical typed $\Diamond$-neighbourhoods}
\begin{lem}\label{CTNSL}
	Let $a=(a_1,a_2)\in W^c$ and suppose that 
	\[
	\delta=\Diamond(\varphi_1,\ldots,\varphi_n;\psi)\in a_2.
	\]
	Then there exists a finite set $A\subseteq W^c$ such that 
	\[
	A\cap\llbracket\psi\rrbracket^c=\varnothing\text{ and } A\not\subseteq\llbracket\varphi_i\rrbracket^c\text{ for every } i\in\{1,\ldots,n\}.
	\]
	\begin{proof}
		If $n=0$, then we may simply take $A=\varnothing$. Then 
		\[
		A\cap\llbracket\psi\rrbracket^c=\varnothing,
		\]
		and the condition 
		\[
		A\not\subseteq\llbracket\varphi_i\rrbracket^c\quad\;(i=1,\ldots,n)
		\]
		is vacuous. So the conclusion is immediate.\\
		\noindent Assume $n\geq 1$. We shall construct, for each $i\in\{1,\ldots,n\}$, a point $b_i\in W^c$ such that 
		\[
		b_i\notin\llbracket\psi\rrbracket^c\text{ and } b_i\notin\llbracket\varphi_i\rrbracket^c.
		\]
		Then we shall set $A=\{b_1,\ldots,b_n\}$.\\
		\noindent Fix $i\in\{1,\ldots,n\}$. \\
		\noindent\textbf{Step 1. Define the relevant generated counter-theory}\\
		Define $F_i=CTh(\{\psi,\varphi_i\})$, the counter-theory generated by $\psi$ and $\varphi_i$. By definition of generated counter-theory, we have 
		\begin{equation}\label{F:1}
		F_i=\{\chi\in\mathcal{L}_{\mathrm{PINL}}: \chi\vdash_{\mathrm{PINL}}\psi\vee\varphi_i\}.
		\end{equation}
		Let $\mathrm{Thm}=\{\zeta\in\mathcal{L}_{\mathrm{PINL}}:\top\vdash_{\mathrm{PINL}}\zeta\}$ denote the collection of theorems in $\mathrm{PINL}$.
		Consider the pair $P_i=(\mathrm{Thm},F_i)$, where $\mathrm{Thm}$ is the theory. We claim that $P_i$ is disjoint. Suppose not. Then by Lemma \ref{DP:1} there exist $\alpha\in\mathrm{Thm}$ and $\beta\in F_i$ such that 
		\[
		\alpha\vdash_{\mathrm{PINL}}\beta.
		\]
		Since $\alpha\in \mathrm{Thm}$, we have 
		\[
		\top\vdash_{\mathrm{PINL}}\alpha.
		\]
		Combining this with $\alpha\vdash_{\mathrm{PINL}}\beta$ and applying $(\mathrm{Cut})$, we obtain
		\[
		\top\vdash_{\mathrm{PINL}}\beta.
		\]
		Now $\beta\in F_i$, so by (\ref{F:1}) we have
		\[
		\beta\vdash_{\mathrm{PINL}}\psi\vee\varphi_i.
		\]
		A further application of $(\mathrm{CUT})$ gives
		\begin{equation}\label{F:2}
		\top\vdash_{\mathrm{PINL}}\psi\vee\varphi_i.
		\end{equation}
		Now $(\Diamond 1)$ yields 
		\[
		\top\vdash_{\mathrm{PINL}}\Diamond(\varphi_1,\ldots,\varphi_{i-1},\top,\varphi_{i+1},\ldots,\varphi_n;\psi).
		\]
		Using $(\Diamond\mathrm{Mon})$ together with $(\ref{F:2})$, we obtain
		\begin{equation}\label{F:3}
		\top\vdash_{\mathrm{PINL}}\Diamond(\varphi_1,\ldots,\varphi_{i-1},\psi\vee\varphi_i,\varphi_{i+1},\ldots,\varphi_n;\psi).
		\end{equation}
		Now applying \((\Diamond2)\) and \((\Diamond5)\) to the formula on the right-hand side of (\ref{F:3}), this yields
		\[
		\top\vdash_{\mathrm{PINL}}\Diamond(\varphi_1,\ldots,\varphi_n;\psi).
		\]
		That is $\top\vdash_{\mathrm{PINL}}\delta$. Thus $\delta$ is a theorem. Since $\top\vdash_{\mathrm{PINL}}\delta$, and since $a_1$ is a theory containing $\top$, closure of $a_1$ under derivability implies that $\delta\in a_1$. But by assumption $\delta\in a_2$. This contradicts the fact that $a=(a_1,a_2)\in W^c$ is a prime theory pair, hence in particular a disjoint pair.\\
		\noindent Therefore our assumption was wrong, and $P_i=(\mathrm{Thm},F_i)$ is disjoint.\\
		Since $P_i$ is disjoint, by Lemma \ref{DP:3} it extends to a prime theory pair
		\[
		b_i=((b_i)_1,(b_i)_2)\in W^c
		\]
		 such that 
		 \[
		 \mathrm{Thm}\subseteq (b_i)_1\text{ and } F_i\subseteq(b_i)_2.
		 \]
		Since $\psi,\varphi_i\in F_i$, we obtain 
		\[
		\psi,\varphi_i\in (b_i)_2.
		\]
		Hence
		\[
		\psi\notin (b_i)_1\text{ and } \varphi_i\notin (b_i)_1.
		\]
		As a result,
		\[
		b_i\notin\llbracket\psi\rrbracket^c\text{ and } b_i\notin\llbracket\varphi_i\rrbracket^c.
		\]
		Now set $A=\{b_1,\ldots,b_n\}$. Then $A\cap\llbracket\psi\rrbracket^c=\varnothing$, and for every $i\in\{1,\ldots,n\}$, since $b_i\in A$ but $b_i\notin\llbracket\varphi_i\rrbracket^c$, hence $A\not\subseteq\llbracket\varphi_i\rrbracket^c$.
			\end{proof}
\end{lem}
By Lemma \ref{CTNSL}, whenever a $\Diamond$-formula 
\[
\delta=\Diamond(\varphi_1,\ldots,\varphi_n;\psi)
\]
belongs to the counter-theoretic component $a_2$ of a canonical state $a=(a_1,a_2)\in W^c$, there exists a finite set $C\subseteq W^c$ such that 
\[
C\cap\llbracket\psi\rrbracket^c=\varnothing
\]
and 
\[
C\nsubseteq\llbracket\varphi_i\rrbracket^c\text{ for every } i=1,\ldots,n.
\]
Equivalently, $C$ fails the co-witness condition associated with the formula $\delta$. We therefore define the canonical typed $\Diamond$-neighbourhood assignment by collecting all formula-labelled pairs $(\delta,C)$ satisfying these conditions. 
\begin{defn}\label{WAL:2}
	For each canonical state $a=(a_1,a_2)\in W^c$, define 
\[
\begin{aligned}
	N_{\tau}^{\Diamond,c}(a)
	:=
	\Bigl\{&
	\bigl(\Diamond(\varphi_{1},\ldots,\varphi_{n};\psi),C\bigr)
	\in
	\operatorname{Form}^{\Diamond}\times\mathcal{P}_{\mathrm{fin}}(W^{c})
	:\;
	\\
	&
	\Diamond(\varphi_{1},\ldots,\varphi_{n};\psi)\in a_{2},
	\\
	&
	C\cap \llbracket\psi\rrbracket^c=\varnothing,
	\\
	&
	C\nsubseteq \llbracket\varphi_i\rrbracket^c
	\quad \text{for every } i=1,\ldots,n
	\Bigr\},
\end{aligned}
\]
where \(\mathcal{P}_{\mathrm{fin}}(W^{c})\) denotes the collection of all
finite subsets of \(W^{c}\).
\end{defn}
\begin{lem}[Persistence of the canonical typed neighbourhood assignments]\label{PTYNA}
	Let \(a,b\in W^{c}\). If \(a\leq^{c} b\), then
	\[
	N_{\tau}^{\Box,c}(a)\subseteq N_{\tau}^{\Box,c}(b)
	\]
	and
	\[
	N_{\tau}^{\Diamond,c}(b)\subseteq N_{\tau}^{\Diamond,c}(a).
	\]
	\begin{proof}
		Let
		\[
		a=(a_{1},a_{2}), \qquad b=(b_{1},b_{2})\in W^{c},
		\]
		be canonical states such that 
		\[
		a\leq^{c} b.
		\]
		By the definition of the canonical order,
		\[
		a_{1}\subseteq b_{1}.
		\]
		Since \(a\) and \(b\) are prime theory pairs, this is equivalently expressed by
		\[
		b_{2}\subseteq a_{2}.
		\]
		
		We first prove the persistence property for the \(\Box\)-neighbourhood assignment. Let
		\[
		\bigl(\Box(\psi_{1},\ldots,\psi_{n};\chi),S\bigr)
		\in N_{\tau}^{\Box,c}(a).
		\]
		By Definition \ref{WAL:1}, we have
		\[
		\Box(\psi_{1},\ldots,\psi_{n};\chi)\in a_{1},
		\]
		and 
		\[
		S\in \mathcal{P}_{\mathrm{fin}}(W^{c}),
		\]
		satisfying
		\[
		S\subseteq \llbracket\chi\rrbracket^c,
		\]
		and
		\[
		S\cap \llbracket\psi_{i}\rrbracket^c\neq\varnothing
		\qquad \text{for every } i=1,\ldots,n.
		\]
		Since \(a_{1}\subseteq b_{1}\), it follows that
		\[
		\Box(\psi_{1},\ldots,\psi_{n};\chi)\in b_{1}.
		\]
		Moreover, the conditions
		\[
		S\subseteq \llbracket\chi\rrbracket^c
		\]
		and
		\[
		S\cap \llbracket\psi_{i}\rrbracket^c\neq\varnothing
		\qquad \text{for every } i=1,\ldots,n
		\]
		depend only on the set \(S\) and on the canonical positive sets
		$\llbracket\chi\rrbracket^c, \llbracket\psi_i\rrbracket^c$ and not on the state
		at which the typed neighbourhood is considered. Therefore,
		\[
		\bigl(\Box(\psi_{1},\ldots,\psi_{n};\chi),S\bigr)
		\in N_{\tau}^{\Box,c}(b).
		\]
		Hence
		\[
		N_{\tau}^{\Box,c}(a)\subseteq N_{\tau}^{\Box,c}(b).
		\]
		We now prove the persistence condition for the \(\Diamond\)-neighbourhood assignment. Let
		\[
		\bigl(\Diamond(\varphi_{1},\ldots,\varphi_{n};\psi),C\bigr)
		\in N_{\tau}^{\Diamond,c}(b).
		\]
		By Definition \ref{WAL:2}, we have
		\[
		\Diamond(\varphi_{1},\ldots,\varphi_{n};\psi)\in b_{2},
		\]
		and 
		\[
		C\in \mathcal{P}_{\mathrm{fin}}(W^{c}),
		\]
		satisfying
		\[
		C\cap \llbracket\psi\rrbracket^c=\varnothing,
		\]
		and
		\[
		C\nsubseteq \llbracket\varphi_i\rrbracket^c
		\qquad \text{for every } i=1,\ldots,n.
		\]
		Since \(b_{2}\subseteq a_{2}\), it follows that
		\[
		\Diamond(\varphi_{1},\ldots,\varphi_{n};\psi)\in a_{2}.
		\]
		Again, the conditions
		\[
	C\cap \llbracket\psi\rrbracket^c=\varnothing,
	\]
	and
	\[
	C\nsubseteq \llbracket\varphi_i\rrbracket^c
	\qquad \text{for every } i=1,\ldots,n.
	\]
		depend only on \(C\) and the canonical positive sets, and not on the
		particular canonical state. Consequently,
		\[
		\bigl(\Diamond(\varphi_{1},\ldots,\varphi_{n};\psi),C\bigr)
		\in N_{\tau}^{\Diamond,c}(a).
		\]
		Therefore,
		\[
		N_{\tau}^{\Diamond,c}(b)\subseteq N_{\tau}^{\Diamond,c}(a).
		\]
		
		Thus both persistence conditions are satisfied.
	\end{proof}
\end{lem}
\subsection{Canonical typed model and Truth Lemma}
Define the canonical typed model by
\[
\mathfrak{M}^{\tau,c}=(W^c,\leq^c,N_\tau^{\Box,c},N_\tau^{\Diamond,c},V^c).
\]
By Lemma \ref{PTYNA}, \(\mathfrak M^{\tau,c}\) is a typed-persistent two-sided neighbourhood model.
\begin{lem}\label{TTL}
	For every formula $\alpha\in\mathcal{L}_{\mathrm{PINL}}$ and every $a\in W^c$,
	\[
	\mathfrak{M}^{\tau,c},\; a\models_\tau\alpha\iff\alpha\in a_1.
	\]
	\begin{proof}
	We prove it by induction on the complexity of $\alpha$. The cases $\alpha=p,\top,\bot,\wedge,\vee$ are standard. For conjunction we use closure of $a_1$ under finite meets; for disjunction we use primeness of $a_1$.\\
	\noindent Now let $\alpha=\Box(\psi_1,\ldots,\psi_n;\chi)$. \\
\noindent First suppose $\alpha\in a_1$. We show $\mathfrak{M}^{\tau,c}, a\models_\tau\alpha$\\
 Since $\alpha\in a_1$, by Lemma \ref{WA}, there exists a finite set $S\subseteq W^c$ such that
	\begin{equation}\label{TTL:1}
	S\subseteq\llbracket\chi\rrbracket^c,\; S\cap\llbracket\psi_i\rrbracket^c\neq\varnothing\text{ for all } i\in\{1,\ldots,n\}.
	\end{equation}
	Since \(\alpha\in a_{1}\), Definition \ref{WAL:1} and (\ref{TTL:1}) yield
	\begin{equation}\label{TTL:2}
	(\alpha,S)\in N_{\tau}^{\Box,c}(a).
	\end{equation}
	By the induction hypothesis, for every $i=1,\ldots,n$ and every $u\in W^c$,
	\[
	\mathfrak{M}^{\tau,c}, u\models_\tau\psi_i\iff \psi_i\in u_1,\quad \mathfrak{M}^{\tau,c}, u\models_\tau\chi\iff\chi\in u_1.
	\]
	Therefore
	\[
	\begin{aligned}
	\llbracket\psi_i\rrbracket^{\tau,c}&=\{u\in W^c:\mathfrak{M}^{\tau,c},u\models_\tau\psi_i\}\\
	&=\{u\in W^c:\psi_i\in u_1\}\\
	&=\llbracket\psi_i\rrbracket^c.
	\end{aligned}
	\]
	and similarly
	\[
	\llbracket\chi\rrbracket^{\tau,c}=\llbracket\chi\rrbracket^c.
	\]
	Hence (\ref{TTL:1}) gives 
	\[
	S\subseteq\llbracket\chi\rrbracket^{\tau,c}, \quad\; S\cap\llbracket\psi_i\rrbracket^{\tau,c}\neq\varnothing\text{ for every } i=1,\ldots,n.
	\]
	Together with (\ref{TTL:2}), this is exactly the typed semantics clause for $\Box$. Thus
	\[
	\mathfrak{M}^{\tau,c}, a\models_\tau\alpha.
	\]
	\noindent Conversely, suppose that $\mathfrak{M}^{\tau,c},a\models_\tau\alpha$. We show $\alpha\in a_1$.\\
\noindent Then by the typed $\Box$-clause there exists \(S\subseteq W^{c}\) such that
	\[
	(\alpha,S)\in N_\tau^{\Box,c}(a)
	\]
	and \(S\) satisfies the required witness conditions. By Definition \ref{WAL:1}, membership of this pair in
	\(N_{\tau}^{\Box,c}(a)\) implies that
	\[
	\alpha\in a_{1}.
	\]
	This proves the \(\Box\)-case.\\
		\noindent Now suppose 
	\[
	\alpha=\Diamond(\varphi_1,\ldots,\varphi_n;\psi)=\delta.
	\]
	\noindent Assume $\delta\in a_1$. Since $a=(a_1,a_2)$ is a prime theory pair, every formula belongs to exactly one of $a_1$ and $a_2$. Hence 
	\[
	\delta\notin a_2.
	\]
	We show that $\mathfrak{M}^{\tau,c}, a\models_\tau\delta$. Assume, for contradiction, that
	\[
	\mathfrak M^{\tau,c},a\not\models_\tau\delta.
	\]
	Then, by the typed \(\Diamond\)-clause, there exists \(C\subseteq W^c\) such that
	\[
	(\delta,C)\in N^{\Diamond,c}_\tau(a)
	\]
	and 
	\[
	\begin{aligned}
	C\cap\llbracket\psi\rrbracket^{\mathfrak M^{\tau,c}}=\emptyset \text{ and }\\
	C\nsubseteq\llbracket\varphi_i\rrbracket^{\mathfrak M^{\tau,c}} \quad\text{for every }i=1,\ldots,n.
\end{aligned}
	\]
	But by Definition \ref{WAL:2}, the membership
	\[
	(\delta,C)\in N^{\Diamond,c}_\tau(a)
	\]
	implies \(\delta\in a_2\). This contradicts \(\delta\notin a_2\). Hence
	\[
	\mathfrak M_{\tau,c},a\models_\tau\delta.
	\]

	\noindent Conversely assume $\mathfrak{M}^{\tau,c}, a\models_\tau\delta$. We show that $\delta\in a_1$\\
	Suppose $\delta\notin a_1$. Since $a=(a_1,a_2)$ is a prime theory pair, it follows that $\delta\in a_2$.
	By Lemma \ref{CTNSL}, there exists a finite set \(A\subseteq W^{c}\) such that
	\begin{equation}\label{TTL:3}
	A\cap\llbracket\psi\rrbracket^{c}=\varnothing,\qquad
	A\nsubseteq\llbracket\varphi_{i}\rrbracket^{c}
	\qquad\text{for every }i=1,\ldots,n.
	\end{equation}
	Since \(\delta\in a_{2}\), Definition \ref{WAL:2} and the equation (\ref{TTL:3}) give
	\begin{equation}\label{TTL:4}
	(\delta,A)\in N_{\tau}^{\Diamond,c}(a).
	\end{equation}

By the induction hypothesis,
	\[
	\llbracket\psi\rrbracket^{\tau,c}=\llbracket\psi\rrbracket^c,\quad\;\llbracket\varphi_i\rrbracket^{\tau,c}=\llbracket\varphi_i\rrbracket^c\; \forall i\in\{1,\ldots,n\}.
	\]
	Hence (\ref{TTL:3}) becomes 
	\[
	A\cap\llbracket\psi\rrbracket^{\tau,c}=\varnothing\text{ and } A\not\subseteq\llbracket\varphi_i\rrbracket^{\tau,c}\;\forall i\in\{1,\ldots,n\}.
	\]
	Together with (\ref{TTL:4}), this contradicts
	\[
	\mathfrak M^{\tau,c},a\models_{\tau}\delta.
	\]
	Therefore
	\[
	\delta\in a_{1}.
	\]
	This completes the induction.

	\end{proof}
\end{lem}
\subsection{Typed completeness of $\mathrm{PINL}$}
\begin{thm}[Completeness with respect to typed persistent semantics]\label{TCOM}
	For all formulas $\varphi,\psi \in \mathcal{L}_{\mathrm{PINL}}$, 
	\[
	\models_\tau(\varphi\vdash\psi)\Rightarrow \varphi\vdash_{\mathrm{PINL}}\psi.
	\]
	\begin{proof}
		We prove the contrapositive. Suppose that
		\begin{equation}\label{TCOM:1}
		\varphi\not\vdash_{\mathrm{PINL}}\psi
		\end{equation}
		Consider the theory pair 
		\[
		\bigl(\operatorname{Th}(\{\varphi\}),
		\operatorname{CTh}(\{\psi\})\bigr).
		\]
		We claim that this pair is disjoint. If not then by Lemma \ref{DP:1} there exist formulas 
		\[
		\alpha\in\operatorname{Th}(\{\varphi\})
		\qquad\text{and}\qquad
		\beta\in\operatorname{CTh}(\{\psi\})
		\]
		such that 
		\begin{equation}\label{TCOM:2}
		\alpha\vdash_{\mathrm{PINL}}\beta.
		\end{equation}
		By the definitions of the generated theory and generated counter-theory,
		\begin{equation}\label{TCOM:3}
		\varphi\vdash_{\mathrm{PINL}}\alpha
		\qquad\text{and}\qquad
		\beta\vdash_{\mathrm{PINL}}\psi.
		\end{equation}
		From (\ref{TCOM:2}) and (\ref{TCOM:3}), two applications of $(\mathrm{CUT})$ yield
		\[
		\varphi\vdash_{\mathrm{PINL}}\psi.
		\]
		contradicting (\ref{TCOM:1}). Hence 
		\[
		\bigl(\operatorname{Th}(\{\varphi\}),
		\operatorname{CTh}(\{\psi\})\bigr)
		\]
		is disjoint. By Lemma \ref{DP:3}, we can extend this disjoint theory pair to a prime theory pair $a=(a_1,a_2)\in W^c$ such that 
		\[
		\operatorname{Th}(\{\varphi\})\subseteq a_{1}
		\qquad\text{and}\qquad
		\operatorname{CTh}(\{\psi\})\subseteq a_{2}.
		\]
		In particular,
		\begin{equation}\label{TCOM:4}
		\varphi\in a_{1}
		\qquad\text{and}\qquad
		\psi\in a_{2}.
		\end{equation}
		Since \(a\) is a prime theory pair, its two components are disjoint. Therefore,
		\begin{equation}\label{TCOM:5}
		\psi\notin a_{1}.
		\end{equation}
		Now we consider the canonical typed model 
		\[
		\mathfrak M^{\tau,c}=(W^c,\leq^c,N_\tau^{\Box,c},N_\tau^{\Diamond,c},V^c).
		\]
	By Lemma \ref{TTL}, and using (\ref{TCOM:4}), (\ref{TCOM:5}) we have 
	\[
	\mathcal M^{\tau,c},a\models_{\tau}\varphi
	\qquad\text{and}\qquad
	\mathcal M^{\tau,c},a\not\models_{\tau}\psi.
	\]
	Hence
		\[
		\llbracket\varphi\rrbracket^{\mathfrak M^{\tau,c}}
		\nsubseteq
		\llbracket\psi\rrbracket^{\mathfrak M^{\tau,c}}.
		\]
		Therefore
		\[
		\mathfrak M^{\tau,c}
		\not\models_{\tau}
		(\varphi\vdash\psi).
		\]
		Since \(\mathfrak M^{\tau,c}\) is a typed persistent two-sided model, it follows that
		\[
		\not\models_{\tau}(\varphi\vdash\psi).
		\]
		Thus
		\[
		\varphi\not\vdash_{\mathrm{PINL}}\psi
		\quad\Longrightarrow\quad
		\not\models_{\tau}(\varphi\vdash\psi).
		\]
		By contraposition,
		\[
		\models_\tau(\phi\vdash\psi)
		\Longrightarrow
		\phi\vdash_{\mathrm{PINL}}\psi.
		\]
		This completes the proof.
	\end{proof}
\end{thm}
		\begin{rem}
			Theorem \ref{TCOM} establishes the completeness of PINL with respect to the auxiliary typed persistent semantics used in the canonical construction. 
			This study introduces typed semantics to establish Lemma \ref{TTL}, but does not offer a general soundness theorem for all typed persistent models.
		\end{rem}

\section{Algebraic Semantics of Positive INL}\label{ASOP}
The algebraic counterpart of positive instantial neighbourhood logic is naturally formulated in terms of distributive lattices equipped with instantial operations. This viewpoint is already suggested by Bezhanishvili et al.\cite{bezhanishvili2020duality} in their coalgebraic treatment of instantial neighbourhood logic.\\
In the present paper, however, the two modalities \(\Box\) and \(\Diamond\) are treated as primitive and independent. We therefore isolate the algebraic content of the two modal fragments before adding any interaction principles between them. This leads to the following notion of a \(2\)-\(\mathrm{DLIO}\): a bounded distributive lattice equipped with two families of finitary operations, one interpreting the \(\Box\)-formulas and the other interpreting the \(\Diamond\)-formulas.\\
The purpose of this section is to formulate the corresponding algebraic semantics and to prove algebraic soundness and completeness for \(\mathrm{PINL}\) with respect to \(2\)-$\mathrm{DLIO}s$.

\begin{defn}\label{DLIO1}
	A $2$-$\mathrm{DLIO}$ is an algebra
	\[
	\mathcal{A}=(A,\wedge,\vee,0,1,(f_n)_{n\in\omega},(g_n)_{n\in\omega})
	\]
	such that $(A,\wedge,\vee,0,1)$ is a bounded distributive lattice and for every $n\in\omega$, the operations
	\[
	f_n,g_n:A^{n+1}\to A
	\]
	satisfy the following conditions.\\
	\noindent \textbf{The $f_n$-laws}:\\
	\noindent For all $a_1,\ldots,a_n,b,c,d,\gamma,\delta\in A$, and every $k\leq n$
	\begin{enumerate}[(F1)]
		\item $f_n(a_1,\ldots,a_{k-1},0,a_{k+1},\ldots,a_n;b)=0$,
		\item  for every permutation $\pi$ of $\{1,\ldots,n\}$, $f_n(a_1,\ldots,a_n;b)=f_n(a_\pi(1),\ldots,a_\pi(n);b)$, 
		\item $f_n(a_1,\ldots,a_{k-1},c\vee d,a_{k+1},\ldots,a_n;b)=f_n(a_1,\ldots,a_{k-1},c,a_{k+1},\ldots,a_n;b)\vee f_n(a_1,\ldots,a_{k-1},d,a_{k+1},\ldots,a_n;b)$,
		\item $f_n(a_1,\ldots,a_n;b\wedge c)\leq f_n(a_1,\ldots,a_n;b)\wedge f_n(a_1,\ldots,a_n;c)$,
		\item $f_n(a_1,\ldots,a_n;b)\leq f_n(a_1,\ldots,a_{n-1},a_n\wedge b;b)$,
		\item if $\gamma\vee\delta=1$, then $f_n(a_1,\ldots,a_n;b)\leq f_{n+1}(a_1,\ldots,a_n,\gamma;b)\vee f_n(a_1,\ldots,a_n;b\wedge \delta)$,
		\item $f_{n+1}(a_1,\ldots,a_n,c;b)\leq f_n(a_1,\ldots,a_n;b)$,
		\item $f_n(a_1,\ldots,a_n;b)\leq f_{n+1}(a_1,\ldots,a_n,a_n;b)$.
	\end{enumerate}
	\textbf{The $g_n$-laws}:
	\begin{enumerate}[(G1)]
		\item $g_n(a_1,\ldots,a_{k-1},1,a_{k+1},\ldots,a_n;b)=1$,
		\item for every permutation $\pi$ of $\{1,\ldots,n\}$, $g_n(a_1,\ldots,a_n;b)=g_n(a_\pi(1),\ldots,a_\pi(n);b)$, 
		\item $g_n(a_1,\ldots,a_{k-1},c\wedge d,a_{k+1},\ldots,a_n;b)=g_n(a_1,\ldots,a_{k-1},c,a_{k+1},\ldots,a_n;b)\wedge g_n(a_1,\ldots,a_{k-1},d,a_{k+1},\ldots,a_n;b)$,
		\item $g_n(a_1,\ldots,a_n;b\wedge c)\leq g_n(a_1,\ldots,a_n;b)\wedge g_n(a_1,\ldots,a_n;c)$,
		\item $g_n(a_1,\ldots,a_{n-1},a_n\vee b;b)\leq g_n(a_1,\ldots,a_n;b)$,
		\item if $\gamma\wedge\delta=0$, then $g_{n+1}(a_1,\ldots,a_n,\gamma;b)\wedge g_n(a_1,\ldots,a_n;b\vee\delta)\leq g_n(a_1,\ldots,a_n;b)$,
		\item $g_n(a_1,\ldots,a_n;b)\leq g_{n+1}(a_1,\ldots,a_n,c;b)$,
		\item $g_{n+1}(a_1,\ldots,a_n,a_n;b)=g_n(a_1,\ldots,a_n;b)$.
	\end{enumerate}
\end{defn}

\subsection*{Derived monotonicity laws}
The syntactic monotonicity rules from Section \ref{PSPINL} are reflected algebraically by the following monotonicity properties of the operations \(f_n\) and \(g_n\).
\begin{lem}\label{MONDLIO}
	In every $2$-$\mathrm{DLIO}$, the following hold.
	\begin{enumerate}
		\item if $a_k\leq a'_k$, then 
		\[
		f_n(a_1,\ldots,a_k,\ldots,a_n;b)\leq f_n(a_1,\ldots, a'_k,\ldots,a_n;b).
		\]
		\item if $b\leq b'$, then 
		\[
		f_n(a_1,\ldots,a_n;b)\leq f_n(a_1,\ldots,a_n;b').
		\]
		\item if $a_k\leq a'_k$, then 
		\[
		g_n(a_1,\ldots,a_k,\ldots,a_n;b)\leq g_n(a_1,\ldots, a'_k,\ldots,a_n;b).
		\]
		\item if $b\leq b'$, then 
		\[
		g_n(a_1,\ldots,a_n;b)\leq g_n(a_1,\ldots,a_n;b').
		\]
	\end{enumerate}
		
	
\end{lem}
\subsection*{Complex algebra of persistent two-sided neighbourhood models}
We now associate a $2$-$\mathrm{DLIO}$ to every persistent two-sided neighbourhood model.\\
Let $\mathfrak M=(W,\leq,N^\Box,N^\Diamond,V)$ be a perisistent model. Then
\[
\operatorname{Up}(W,\leq)=\{U\subseteq W: U \text{ is upward closed with respect to }\leq\}.
\]
For $U_1,\ldots, U_n,V\in \operatorname{Up}(W,\leq)$, define
\[
f_n^\mathfrak{M}(U_1,\ldots,U_n;V)=\{w\in W:\exists S\in N^\Box(w)\text{ s.t. } S\subseteq V \text{ and } S\cap U_i\neq\varnothing \forall i=1,\ldots,n\}.
\]
and 

\[
g_n^\mathfrak{M}(U_1,\ldots, U_n;V)=\Bigl\{w\in W:\forall T\in N^\Diamond(w),
	\bigl[(T\subseteq U_i\text{ for some }i=1,\ldots,n)	\text{ or }(T\cap V\neq\emptyset)\bigr]
	\Bigr\}.
\]
\begin{lem}\label{CALG}
	For all \(U_1,\ldots,U_n,V\in\operatorname{Up}(W,\leq)\),
	\[
	f_n^{\mathfrak M}(U_1,\ldots,U_n;V)
	\quad\text{and}\quad
	g_n^{\mathfrak M}(U_1,\ldots,U_n;V)
	\]
	are up-sets of \((W,\leq)\).
	
	\begin{proof}
	For \(f_n^{\mathfrak M}\), if $w\in f_n^\mathfrak{M}(U_1,\ldots,U_n;V)$ and $w\leq v$, then a witness  $S\in N^\Box(w)$ is also in \(N^\Box(v)\), since
		\[
		N^\Box(w)\subseteq N^\Box(v).
		\]
		Hence $v\in f_n^\mathfrak{M}(U_1,\ldots,U_n;V)$. Thus $f_n^\mathfrak{M}(U_1,\ldots,U_n;V)$ is an upset.\\
	\noindent Suppose $w\in g_n^\mathfrak{M}(U_1,\ldots, U_n;V)$ and $w\leq v$. Let $T\in N^\Diamond(v)$ be arbitrary. Since the frame is persistent and \(w\leq v\), we have \[N^\Diamond(v)\subseteq N^\Diamond(w)\]. Hence $T\in N^\Diamond(w)$. Since $w\in g_n^\mathfrak{M}(U_1,\ldots, U_n;V)$, it follows that either
	\[
	T\subseteq U_i
	\quad\text{for some }i=1,\ldots,n,
	\]
	or 
	\[
	T\cap V\neq\emptyset.
	\]
	Since \(T\in N^\Diamond(v)\) was arbitrary, this proves $v\in g_n^\mathfrak{M}(U_1,\ldots, U_n;V)$. Therefore, $g_n^\mathfrak{M}(U_1,\ldots, U_n;V)$ is an upset.
	\end{proof}
\end{lem}
\begin{prop}\label{EPD}
	For every persistent two-sided neighbourhood model $\mathfrak M=(W,\leq,N^\Box,N^\Diamond,V)$, the complex algebra $\mathbb{A}(\mathfrak M)=(\operatorname{Up}(W,\leq),\cap,\cup,\emptyset,W,(f^\mathfrak{M}_n)_{n\in \omega},(g^\mathfrak{M}_n)_{n\in\omega})$ of $\mathfrak M$ is a $2$-DLIO.
	\begin{proof}
		The lattice reduct
		\[
		(\operatorname{Up}(W,\leq),\cap,\cup,\emptyset,W)
		\]
		is a bounded distributive lattice. By Lemma \ref{CALG}, the operations
		\[
		f_n^{\mathfrak M},g_n^{\mathfrak M}
		\]
		are well-defined on \(\operatorname{Up}(W,\leq)\).
		
		It remains to verify the laws given in Definition \ref{DLIO1}. Each of these laws is the algebraic counterpart of the corresponding sound modal axiom from Section \ref{PSPINL}. Since valuations in persistent models assign arbitrary up-sets to propositional variables, the validity of those modal axioms yields the corresponding equations and inequations for arbitrary elements of \(\operatorname{Up}(W,\leq)\). Equivalently, these laws follow directly from the defining clauses of \(f_n^{\mathfrak M}\) and \(g_n^{\mathfrak M}\).
Thus all defining conditions of a \(2\)-\(\mathrm{DLIO}\) are satisfied. Thus $\mathbb{A}(\mathfrak M)$ is a $2$-$\mathrm{DLIO}$.
	\end{proof}
\end{prop}
\subsection*{Algebraic valuations and validity}
\begin{defn}
	Let $\mathbb A=(A,\wedge,\vee,0,1,(f_n){n\in\omega},(g_n){n\in\omega})$ be a $2$-$\mathrm{DLIO}$. A valuation in $\mathbb{A}$ is a map
	\[
	v:\mathrm{Prop}\to A.
	\]
	It extends uniquely to all positive INL formulas by the following clauses:
	\[
	\begin{aligned}
		v(\top)=1,\quad\; v(\bot)=0,\\
		v(\varphi\wedge\psi)= v(\varphi)\wedge v(\psi),\quad v(\varphi\vee\psi)= v(\varphi)\vee v(\psi),\\
		v(\Box(\varphi_1,\ldots,\varphi_n;\psi))= f_n(v(\varphi_1),\ldots,v(\varphi_n);v(\psi)),\\
		v(\Diamond(\varphi_1,\ldots,\varphi_n;\psi))= g_n(v(\varphi_1),\ldots,v(\varphi_n);v(\psi)).
	\end{aligned}
	\]
	For \(n=0\), these clauses read
	\[
	v(\Box(\psi))=f_0(v(\psi)),
	\qquad
	v(\Diamond(\psi))=g_0(v(\psi)).
	\]
\end{defn}
\begin{defn}
		Let $\mathbb A=(A,\wedge,\vee,0,1,(f_n){n\in\omega},(g_n){n\in\omega})$ be a $2$-$\mathrm{DLIO}$. The algebraic validity of a sequent $\varphi\vdash\psi$ is defined by:
		\[
		\mathbf{A}\models (\varphi\vdash\psi)
		\]
		if and only if for every valuation $v:\mathrm{Prop}\to A$,
		\[
		v(\varphi)\leq v(\psi).
		\]
\end{defn}
It is valid in the class of all $2$-$\mathrm{DLIOs}$ if it is valid in every $2$-$\mathrm{DLIO}$.
\subsection{Algebraic soundness and completeness of PINL}
\begin{thm}
	Let $\mathbf{A}=(A,\wedge,\vee,0,1,(f_n){n\in\omega},(g_n){n\in\omega})$ be a $2$-$\mathrm{DLIO}$. If 
	\[
	\varphi\vdash_{\mathrm{PINL}}\psi,
	\]
	then $\mathbf{A}\models (\varphi\vdash\psi)$.
	\begin{proof}
		Let \(v:\mathrm{Prop}\to A\) be an arbitrary algebraic valuation, extended recursively to all \(\mathrm{PINL}\)-formulas. We prove the result by induction on derivations.
		The propositional axioms are valid because
		\[
		(A,\wedge,\vee,0,1)
		\]
		is a bounded distributive lattice. The modal axiom schemata are valid by the defining \(2\)-\(\mathrm{DLIO}\) laws. In particular, the \(\Box\)-cover axiom $\Box6$ is interpreted by the corresponding \(f_n\)-cover law, using the side condition \(v(\gamma)\vee v(\delta)=1\); similarly, the dual \(\Diamond\)-cover axiom $\Diamond6$ is interpreted by the corresponding \(g_n\)-cover law, using the side condition \(v(\gamma)\wedge v(\delta)=0\).
		It remains to note that the inference rules preserve algebraic validity. The rule \((CUT)\) follows from transitivity of the lattice order. The rule \((US)\) is sound because, for every substitution \(\sigma\), the map
		\[
		v_\sigma(p)=v(\sigma(p))
		\]
		is again an algebraic valuation, and a routine induction gives
		\[
		v(\sigma(\theta))=v_\sigma(\theta)
		\]
		for every formula \(\theta\). The rule \((RE)\) is sound because the operations \(\wedge,\vee,f_n,g_n\) are extensional: if \(v(\alpha)=v(\beta)\), then replacing \(\alpha\) by \(\beta\) inside any formula does not change its value. Finally, the monotonicity rules \((\Box\mathrm{Mon})\) and \((\Diamond\mathrm{Mon})\) preserve validity by Lemma \ref{MONDLIO}.
		Hence every derivable sequent is valid in \(\mathbb A\). Since \(\mathbb A\) was arbitrary,
		\[
		\varphi\vdash_{\mathrm{PINL}}\psi
		\quad\Longrightarrow\quad
		\mathbb A\models(\varphi\vdash\psi).
		\]
	\end{proof}
\end{thm}

\subsection*{The Lindenbaum Algebra of Positive INL}
We now construct the Lindenbaum algebra of \(\mathrm{PINL}\) by identifying provably equivalent formulas. Define
\[
\varphi\equiv_{\mathrm{PINL}}\psi
\quad\Longleftrightarrow\quad
\varphi\vdash_{\mathrm{PINL}}\psi
\ \text{ and }
\psi\vdash_{\mathrm{PINL}}\varphi.
\]
Let \([\varphi]\) denote the equivalence class of \(\varphi\). Thus
\[
\mathcal L_{\mathrm{PINL}}/\equiv_{\mathrm{PINL}}
\]
is the set of all equivalence classes of \(\mathrm{PINL}\)-formulas under provable equivalence.

\begin{defn}
	The Lindenbaum algebra of $\mathrm{PINL}$ is the quotient
	\[
	\mathfrak{L}_{\mathrm{PINL}}=(\mathcal{L}_{\mathrm{PINL}}/\equiv_{\mathrm{PINL}},\wedge,\vee,[\bot],[\top],(f_n^\mathfrak{L})_{n\in\omega},(g_n^\mathfrak{L})_{n\in\omega}),
	\]
	where 
	\[
	\begin{aligned}
	[\varphi]\wedge [\psi]=[\varphi\wedge\psi], [\varphi]\vee[\psi]=[\varphi\vee\psi],\\
	f_n^\mathfrak{L}([\varphi_1],\ldots,[\varphi_n];[\psi])=[\Box(\varphi_1,\ldots,\varphi_n;\psi)],\\
	g_n^\mathfrak{L}([\varphi_1],\ldots,[\varphi_n];[\psi])=[\Diamond(\varphi_1,\ldots,\varphi_n;\psi)].
	\end{aligned}
	\]
	For \(n=0\), these clauses read
	\[
	f_0^L([\psi])=[\Box(\psi)],
	\qquad
	g_0^L([\psi])=[\Diamond(\psi)].
	\]
\end{defn}
 \begin{lem}\label{WDLIO}
 	The operations $f_n^\mathfrak{L}$ and $g_n^\mathfrak{L}$ are well-defined.
 	\begin{proof}
 	Suppose $[\varphi_i]=[\varphi_i']$ for $i=1,\ldots,n$ and $[\psi]=[\psi']$. Then
 		\[
 		\varphi_i\equiv_{\mathrm{PINL}}\varphi_i'\;(i=1,\ldots, n)\;\psi\equiv_{\mathrm{PINL}}\psi'.
 		\] 
 		Hence 
 		\[
 		\varphi_i\vdash_{\mathrm{PINL}}\varphi_i'\text{ and } \varphi_i'\vdash_{\mathrm{PINL}}\varphi_i\;(i=1,\ldots,n),
 		\]
 		and 
 		\[
 		\psi\vdash_{\mathrm{PINL}}\psi'\text{ and } \psi'\vdash_{\mathrm{PINL}}\psi.
 		\]
 		Using $(\Box\mathrm{Mon})$, we obtain
 		\[
 		\begin{aligned}
 			\Box(\varphi_1,\ldots,\varphi_n;\psi)\vdash_{\mathrm{PINL}}\Box(\varphi_1',\ldots,\varphi_n';\psi')\\
 			\Box(\varphi_1',\ldots,\varphi_n';\psi')\vdash_{\mathrm{PINL}}\Box(\varphi_1,\ldots,\varphi_n;\psi).
 		\end{aligned}
 		\]
 		Therefore 
 		\[
 		\Box(\varphi_1,\ldots,\varphi_n;\psi)\equiv_{\mathrm{PINL}} 	\Box(\varphi_1',\ldots,\varphi_n';\psi'),
 		\]
 		and hence
 		\[
 		[	\Box(\varphi_1,\ldots,\varphi_n;\psi)]=[\Box(\varphi_1',\ldots,\varphi_n';\psi')].
 		\]
 		Thus $f_n^{\mathfrak{L}}$ is well-defined.\\
 		The well-definedness of the operation $g_n^\mathfrak{L}$ is proved exactly the same way, by applying $(\Diamond\mathrm{Mon})$.
 		\end{proof}
 	\end{lem}
 \begin{prop}\label{LG}
 $\mathfrak{L}_{\mathrm{PINL}}$ is a $2$-$\mathrm{DLIO}$.
 \begin{proof}
 	The lattice reduct of $\mathfrak{L}_{\mathrm{PINL}}$ is $(\mathcal{L}_{\mathrm{PINL}}/\equiv_{\mathrm{PINL}},\wedge,\vee,[\bot],[\top])$. This is the usual Lindenbaum algebra of the distributive-lattice fragment, and hence it is a bounded distributive lattice. The operations
 	\[
 	f_n^{\mathfrak L}
 	\quad\text{and}\quad
 	g_n^{\mathfrak L}
 	\qquad(n\in\omega)
 	\]
 	are well-defined by Lemma \ref{WDLIO}. It remains only to verify the defining \(2\)-\(\mathrm{DLIO}\) laws. Each law is exactly the algebraic form of the corresponding modal axiom of \(\mathrm{PINL}\). For example, if $[\gamma]\vee[\delta]=[\top]$, then $\top\dashv\vdash_{\mathrm{PINL}}\gamma\vee\delta$. By the \(\Box\)-cover axiom $(\Box6)$,
 	\[
 	\Box(\varphi_1,\ldots,\varphi_n;\psi)
 	\vdash_{\mathrm{PINL}}
 	\Box(\varphi_1,\ldots,\varphi_n,\gamma;\psi)
 	\vee
 	\Box(\varphi_1,\ldots,\varphi_n;\psi\wedge\delta),
 	\]
 	which gives 
 	 \[
 	f_n^\mathfrak{L}([\varphi_1],\ldots,[\varphi_n];[\psi])\leq f_{n+1}^\mathfrak{L}([\varphi_1],\ldots,[\varphi_n],[\gamma];[\psi])\vee f_n^\mathfrak{L}([\varphi_1],\ldots,[\varphi_n];[\psi]\wedge[\delta]),
 	\]
 	the corresponding \(f_n^{\mathfrak L}\)-cover law $(F6)$ in \(\mathfrak{L}_{\mathrm{PINL}}\). The dual \(g_n^{\mathfrak{L}}\)-cover law \((G6)\) is obtained similarly from the \(\Diamond\)-cover axiom \((\Diamond6)\), using the condition $[\gamma]\wedge[\delta]=[\bot]$. All remaining \(f_n^{\mathfrak{L}}\)- and \(g_n^{\mathfrak L}\)-laws follow in the same way from the corresponding \(\Box\)- and \(\Diamond\)-axioms of \(\mathrm{PINL}\). Therefore $\mathfrak L_{\mathrm{PINL}}$ is a \(2\)-\(\mathrm{DLIO}\).

 \end{proof}
\end{prop}
\begin{lem}\label{OLA}
Let $\leq_{\mathfrak{L}_{\mathrm{PINL}}}$ denote the lattice order of the bounded distributive lattice reduct of $\mathfrak{L}_{\mathrm{PINL}}$. For all formulas $\alpha,\beta\in\mathcal{L}_{\mathrm{PINL}}$,
	\[
	[\alpha]\leq_{\mathfrak{L}_{\mathrm{PINL}}}[\beta]\iff\alpha\vdash_{\mathrm{PINL}}\beta.
	\]
	\begin{proof}
		In the lattice reduct of $\mathfrak{L}_{\mathrm{PINL}}$, the order is given by
		\[
		x\leq y\iff x\wedge y=x.
		\]
		Hence
		\[
		[\alpha]\leq_{\mathfrak{L}_{\mathrm{PINL}}}[\beta]\iff [\alpha]\wedge[\beta]=[\alpha].
		\]
		Since 
		\[
		[\alpha]\wedge[\beta]=[\alpha\wedge\beta], 
		\]
		this is equivalent to $[\alpha\wedge\beta]=[\alpha]$,
		That is
		\[
		[\alpha]\leq_{\mathfrak{L}_{\mathrm{PINL}}}[\beta]\iff\alpha\wedge\beta\equiv_{\mathrm{PINL}}\alpha.
		\]
		We shall show that 
		\[
		\alpha\wedge\beta\equiv_{\mathrm{PINL}}\alpha\iff \alpha\vdash_{\mathrm{PINL}}\beta.
		\]
		First consider $\alpha\vdash_{\mathrm{PINL}}\beta$. \\
		By reflexivity, $\alpha\vdash_{\mathrm{PINL}}\alpha$. Using conjunction introduction, we get
		\[
		\alpha\vdash_{\mathrm{PINL}}\alpha\wedge\beta.
		\]
		Also, by conjunction elimination, 
		\[
		\alpha\wedge\beta\vdash_{\mathrm{PINL}}\alpha.
		\]
		Hence $\alpha\wedge\beta\equiv_{\mathrm{PINL}}\alpha$. Therefore, $[\alpha]\leq_{\mathfrak{L}_{\mathrm{PINL}}}[\beta]$.\\
		Conversely, suppose 
		\[
		\alpha\wedge\beta\equiv_{\mathrm{PINL}}\alpha.
		\]
		Then, in particular, $\alpha\vdash_{\mathrm{PINL}}\alpha\wedge\beta$. By conjunction elimination, 
		\[
		\alpha\wedge\beta\vdash_{\mathrm{PINL}}\beta.
		\]
		Therefore, by $(\mathrm{CUT})$, $\alpha\vdash_{\mathrm{PINL}}\beta$. Thus
		\[
		[\alpha]\leq_{\mathfrak{L}_{\mathrm{PINL}}}[\beta]\iff\alpha\vdash_{\mathrm{PINL}}\beta.
		\]
	\end{proof}
	
\end{lem}

\subsection*{Algebraic Completeness}
The following algebraic completeness theorem is obtained by the standard Lindenbaum-algebra argument. Its role here is to show that the algebraic semantics based on \(2\)-\(\mathrm{DLIO}\)s exactly captures the proof system of \(\mathrm{PINL}\).
 \begin{thm}
 	For all $\mathrm{PINL}$ formulas $\varphi,\psi$, 
 	\[
 	\models_{\text{$2$-$\mathrm{DLIO}$}} (\varphi\vdash\psi)\Rightarrow \varphi\vdash_{\mathrm{PINL}}\psi,
 	\]
 	where $\models_{\text{$2$-$\mathrm{DLIO}$}} (\varphi\vdash\psi)$ means validity of the sequent $\varphi\vdash\psi$ in every $2$-$\mathrm{DLIO}$ under every algebraic valuation.
 	\begin{proof}
 		Suppose $\varphi\not\vdash_{\mathrm{PINL}}\psi$. Consider the Lindenbaum algebra $\mathfrak{L}_{\mathrm{PINL}}=\mathcal{L}_{\mathrm{PINL}}/\equiv_{\mathrm{PINL}}$. By Proposition \ref{LG}, $\mathfrak{L}_{\mathrm{PINL}}$ is a $2$-$\mathrm{DLIO}$. Let $v_0$ be the canonical valuation defined by 
 		\[
 		v_0(p)=[p].
 		\]
 		By induction on formulas, we obtain 
 		\[
 		v_0(\theta)=[\theta], 
 		\]
 		for every $\mathrm{PINL}$-formula $\theta$. Hence 
 		\[
 		v_0(\varphi)=[\varphi], v_0(\psi)=[\psi].
 		\]
 	By Lemma \ref{OLA},
 		\[
 		[\alpha]\leq_{\mathfrak{L}_{\mathrm{PINL}}}[\beta]\iff\alpha\vdash_{\mathrm{PINL}}\beta.
 		\]
 		Since $\varphi\not\vdash_{\mathrm{PINL}}\psi$, we have 
 		\[
 		[\varphi]\not\leq_{\mathfrak{L}_{\mathrm{PINL}}}[\psi].
 		\]
 		Therefore
 		\[
 		v_0(\varphi)\not\leq_{\mathfrak{L}_{\mathrm{PINL}}} v_0(\psi).
 		\]
 		Thus the sequent $\varphi\vdash\psi$ is not valid in all $2$-$\mathrm{DLIOs}$.  Hence
 		\[
 		\models_{2\text{-}\mathrm{DLIO}}(\varphi\vdash\psi)
 		\quad\Longrightarrow\quad
 		\varphi\vdash_{\mathrm{PINL}}\psi.
 		\]
 	\end{proof}
 \end{thm}
\section{Bitopological PINL-spaces and Admissible-open Representation}\label{BPAOR}
The typed canonical semantics introduced in Section \ref{CPMTC} was used to prove completeness.
Bitopological spaces, introduced by Kelly \cite{kelly1963bitopological}, are spaces equipped with two topologies. In the present canonical construction, these two topologies are generated by the positive and negative opens determined by prime theory pairs.
However, the typed canonical model is still partly syntactic, because its neighbourhoods are labelled by formulas. \\
For a representation-theoretic treatment, this syntactic dependence should be removed. If two formulas are provably equivalent, then they determine the same element of the Lindenbaum $2$-$\mathrm{DLIO}$, and therefore they should determine the same data on the space side. The aim of this section is to replace formula labels by admissible positive opens and to show that the resulting canonical bitopological structure recovers the Lindenbaum $2$-$\mathrm{DLIO}$.

\subsection{The canonical bitopological setting for Positive INL}
Let $W^c$ be the set of prime theory pairs $a=(a_1,a_2)$. For each formula $\varphi$, define
\[
U_\varphi=\{a\in W^c:\varphi\in a_1\},
\]
and 
\[
V_\varphi=\{a\in W^c:\varphi\in a_2\}.
\]
For each formula $\varphi$ in positive instantial neighbourhood language, the set $U_\varphi$ consists of those prime pairs whose positive component contains $\varphi$, while $V_\varphi$ consists of those prime pairs whose negative component contains $\varphi$.

\begin{defn}
	Define two topologies on $W^c$ by 
	\[
	\tau^+=\text{ the topology generated by } \{U_\varphi:\varphi\in \mathrm{Fm}\},
	\]
	and 
	\[
	\tau^-=\text{ the topology generated by } \{V_\varphi:\varphi\in\mathcal{L}_{\mathrm{PINL}}\}.
	\]
	Then $(W^c,\tau^+,\tau^-)$ is a bitopological space.
\end{defn}
\begin{note}\label{ACDLIO}
The collection $\{U_\varphi:\varphi\in \mathcal{L}_{\mathrm{PINL}}\}$ forms the lattice of admissible positive opens, and this lattice will serve as the underlying lattice for the algebraic semantics.\\
We call $A^c=\{U_\varphi:\varphi\in\mathcal{L}_{\mathrm{PINL}}\}$ the lattice of admissible positive opens. It is closed under finite unions and finite intersections, since
\[
\begin{aligned}
U_{\varphi\wedge\psi}=U_\varphi\cap U_\psi,\\
U_{\varphi\vee\psi}=U_\varphi\cup U_\psi,
\end{aligned}
\]
and
\[
U_\bot=\varnothing,\; U_\top=W^c.
\]
Hence, $A^c$ satisfies the basis condition and also $(A^c,\cap,\cup,\emptyset,W^c)$ is a bounded distributive lattice of admissible positive opens.\\
Let $B^c=\{V_\varphi:\varphi\in\mathcal{L}_{\mathrm{PINL}}\}$. Now,
\[
V_\bot=W^c,
\]
because $\bot\in a_2$ for every prime counter-theory $a_2$. Hence $B^c$ covers $W^c$. Since $V_{\varphi\vee\psi}=V_\varphi\cap V_\psi$, $B^c$ is closed under finite intersection. Therefore, $B^c$ is a basis for $\tau^-$. \\
We have 
\[
V_{\varphi\wedge \psi}=V_\varphi\cup V_\psi,\;V_{\varphi\vee\psi}=V_\varphi\cap V_\psi, \; V_\bot=W^c,\; V_\top=\emptyset.
\]
Thus $B^c$ is closed under finite intersections and finite unions and contains the bottom and top elements. Hence $(B^c,\cap,\cup,\emptyset,W^c)$ is a bounded distributive lattice.
\end{note}
\subsection*{The separation lemma for admissible opens}
\begin{lem}\label{SLAO}
	For all formulas $\varphi,\psi$ in the language of $\mathrm{PINL}$, 
	\[
	U_\varphi\subseteq U_\psi
	\]
	if and only if
	\[
	\varphi\vdash_{\mathrm{PINL}}\psi.
	\]
	Consequently,
	\[
	U_\varphi=U_\psi
	\]
	if and only if
	\[
	\varphi\dashv\vdash_{\mathrm{PINL}}\psi.
	\]
	\begin{proof}
		Suppose $\varphi\vdash_{\mathrm{PINL}}\psi$. Let $a=(a_1,a_2)\in U_\varphi$. Then $\varphi\in a_1$. Since $a_1$ is a theory and is closed under provable consequence, $\psi\in a_1$. Hence $a\in U_\psi$. Therefore $U_\varphi\subseteq U_\psi$.\\
		Conversely, suppose that $U_\varphi\subseteq U_\psi$. We prove that $\varphi\vdash_{\mathrm{PINL}}\psi$. Assume, for contradiction, that 
		\[
		\varphi\not\vdash_{\mathrm{PINL}}\psi.
		\]
		Then $Th(\{\varphi\})$, the theory generated by $\varphi$ and $CTh(\{\psi\})$, the counter-theory generated by $\psi$ form a disjoint theory pair. Then by Lemma \ref{DP:3}, there exists a prime theory pair
		\[
		a=(a_1,a_2)\in W^c
		\]
		such that 
		\[
		\varphi\in a_1\;\text{ and }\;\psi\in a_2.
		\]
		Since the pair is disjoint, $\psi\notin a_1$. Hence $a\in U_\varphi$, but $a\notin U_\psi$. This contradicts $U_\varphi\subseteq U_\psi$. Therefore $\varphi\vdash_{\mathrm{PINL}}\psi$.\\
		The equivalence $U_\varphi=U_\psi$ if and only if $\varphi\dashv\vdash_{\mathrm{PINL}}\psi$ follows immediately by applying the first part in both directions.
	\end{proof}
\end{lem}
\subsection{Typed bitopological $\mathrm{PINL}$-spaces}
The prime-pair typed canonical model constructed in Section \ref{CPMTC} suggests the data needed on space side. Since canonical states contain both positive and counter-theoretic information, we start with a bitopological space $(X,\tau^+,\tau^-)$.  We need a distinguished bounded distributive lattice $\mathcal{A}\subseteq\tau^+$ of admissible positive opens, since the \(2\)-\(\mathrm{DLIO}\) operations will be interpreted on these opens. Finally, we need typed neighbourhood assignments $N_\tau^\Box$ and $N_\tau^\Diamond$. These are required because the instantial operators are not determined by the two topologies alone. They also depend on the neighbourhood data associated with each tuple $(A_1,\ldots,A_n;B)$ of admissible positive opens. 
The labels are therefore taken from $\mathcal{A}$, rather than from formulas. This makes the structure depend only on the algebraic meaning of formulas: by Lemma \ref{SLAO}, provably equivalent formulas determine the same admissible positive open. This motivates the following definition.

\begin{defn}\label{TBPS}
	A typed bitopological PINL-space is a structure 
	\[
	\mathfrak{X}=(X,\tau^+,\tau^-,\mathcal{A},N_\tau^\Box,N_\tau^\Diamond)
	\]
	such that 
	\begin{enumerate}[(i)]
		\item $(X,\tau^+,\tau^-)$ is a bitopological space;
		\item $\mathcal{A}\subseteq\tau^+$ is a bounded distributive lattice of admissible positive opens, with bottom $\emptyset$ and top $X$;
		\item for each $x\in X$, 
		\[
		N_\tau^\Box(x)\subseteq\bigcup_{n\in\omega}(\mathcal{A}^{n+1}\times \mathcal{P}(X));
		\]
		\item for each $x\in X$, 
		\[
		N_\tau^\Diamond(x)\subseteq\bigcup_{n\in\omega }(\mathcal{A}^{n+1}\times\mathcal{P}(X)).
		\]
	\end{enumerate}
	\end{defn}
An element of $N_\tau^\Box(x)$ has the form
\[
((A_1,\ldots,A_n;B),S),
\]
where $A_1,\ldots,A_n, B\in \mathcal{A}$ and $S\subseteq X$. The tuple 
\[
(A_1,\ldots,A_n;B)
\]
is the admissible-open label attached to the neighbourhod $S$. Similarly, an elements of $N_\tau^\Diamond(x)$ has the form
\[
((A_1,\ldots,A_n;B),C),
\]
where $C\subseteq X$, is the neighbourhood set, while the tuple $(A_1,\ldots,A_n;B)$ is the admissible-open label attached to $C$.
\subsection*{The algebra of admissible opens}
Let $\mathfrak{X}=(X,\tau^+,\tau^-,\mathcal{A},N_\tau^\Box,N_\tau^\Diamond)$ be a typed bitopological PINL-space.\\
For $A_1,\ldots,A_n,B\in \mathcal{A}$, define 
\[
f_n^\mathfrak{X}(A_1,\ldots,A_n;B)
\]
to be the set of all $x\in X$ such that there exists $S\subseteq X$ satisfying 
\[
((A_1,\ldots,A_n;B),S)\in N_\tau^\Box(x),\\
\]
\[
S\subseteq B,\quad \text{ and }\quad S\cap A_i\neq\emptyset\;\text{ for every } i=1,\ldots,n.
\]
Thus $f_n^\mathfrak{X}$ is the operation induced by the typed $\Box$-clause. \\
Similarly, define $g_n^\mathfrak{X}(A_1,\ldots,A_n;B)$ to be the set of all $x\in X$ such that there is no $C\subseteq X$ satisfying 
\[
((A_1,\ldots,A_n;B),C)\in N_\tau^\Diamond(x),\\
\]
\[
C\cap B=\emptyset,\text{ and } C\not\subseteq A_i\;\text{ for every } i=1,\ldots,n.
\]
Equivalently, $x\in g_n^\mathfrak{X}(A_1,\ldots,A_n;B)$ if and only if for every $C\subseteq X$ with 
\[
((A_1,\ldots,A_n;B),C)\in N_\tau^\Diamond(x)
\]
satisfies 
\[
C\cap B\neq\emptyset, \text{ or } C\subseteq A_i\;\text{ for some } i=1,\ldots,n.
\]
Thus $g_n^\mathfrak{X}$ is the operation induced by the typed $\Diamond$-clause.
\begin{defn}\label{TBOC}
	A typed bitopological PINL-space $\mathfrak{X}=(X,\tau^+,\tau^-,\mathcal{A},N_\tau^\Box,N_\tau^\Diamond)$ is called operation-closed if, for every \(n\in\omega\) and for all $A_1,\ldots,A_n,B\in\mathcal{A}$, we have
	\[
	f_n^\mathfrak{X}(A_1,\ldots,A_n;B)\in \mathcal{A}
	\]
	and 
	\[
	g_n^\mathfrak{X}(A_1,\ldots,A_n;B)\in\mathcal{A}.
	\]
	It is called algebraically admissible if it is operation-closed and the induced operations
	\[
	(f_n^\mathfrak{X})_{n\in\omega},\; (g_n^\mathfrak{X})_{n\in\omega}
	\]
	satisfy the defining $2$-$\mathrm{DLIO}$ laws.
\end{defn}
\begin{prop}\label{TBPA2}
	If the typed bitopological PINL-space $\mathfrak{X}=(X,\tau^+,\tau^-,\mathcal{A},N_\tau^\Box,N_\tau^\Diamond)$ is algebraically admissible, then 
	\[
	\mathbb{A}(\mathfrak{X})=(\mathcal{A},\cap,\cup,\emptyset,X,(f_n^\mathfrak{X})_{n\in\omega}, (g_n^\mathfrak{X})_{n\in\omega})
	\]
	is a $2$-$\mathrm{DLIO}$.
	\begin{proof}
		By Definition \ref{TBPS}, $(\mathcal{A},\cap,\cup,\emptyset,X)$ is a bounded distributive lattice of admissible positive opens.
		Since $\mathfrak X$ is algebraically admissible, the operations $f_n^\mathfrak{X}$ and $g_n^\mathfrak{X}$ are closed on $\mathcal{A}$ and satisfy the defining $2$-$\mathrm{DLIO}$ laws. Hence $\mathbb{A}(\mathfrak{X})$ is a $2$-$\mathrm{DLIO}$.
	\end{proof}
\end{prop}
\subsection{The canonical bitopological PINL-space}
We now define a bitopological PINL-space from the prime-pair typed canonical model developed in Section \ref{CPMTC}. First we require the following basic fact.
\begin{cor}[Independence of modal open labels from formula representatives]\label{IRC}
	Let $n\in\omega$, and let
	\[
	\varphi_1,\ldots,\varphi_n,\theta_1,\ldots,\theta_n,\chi,\rho\in\mathcal{L}_{\mathrm{PINL}}.
	\]
	Suppose that 
	\[
	U_{\varphi_i}=U_{\theta_i}\text{ for every } i=1,\ldots,n,
	\]
	and
	\[
	U_\chi=U_\rho.
	\]
	Then
	\[
	\Box(\varphi_1,\ldots,\varphi_n;\chi)\dashv\vdash_{\mathrm{PINL}}\Box(\theta_1,\ldots,\theta_n;\rho),
	\]
	and
	\[
	\Diamond(\varphi_1,\ldots,\varphi_n;\chi)\dashv\vdash_{\mathrm{PINL}}\Diamond(\theta_1,\ldots,\theta_n;\rho).
	\]
	Consequently, for every prime theory pair $a=(a_1,a_2)\in W^c$,
	\[
	\Box(\varphi_1,\ldots,\varphi_n;\chi)\in a_1\iff\Box(\theta_1,\ldots,\theta_n;\rho)\in a_1,
	\]
	and
	\[
	\Diamond(\varphi_1,\ldots,\varphi_n;\chi)\in a_2\iff\Diamond(\theta_1,\ldots,\theta_n;\rho)\in a_2.
	\]
	\begin{proof}
		By Lemma \ref{SLAO}, the equality $U_{\varphi_i}=U_{\theta_i}$ implies 
		\[
		\varphi_i\dashv\vdash_{\mathrm{PINL}}\theta_i\text{ for every } i=1,\ldots,n.
		\]
		Similarly, $U_\chi=U_\rho$ implies
		\[
		\chi\dashv\vdash_{\mathrm{PINL}}\rho.
		\]
		By repeated applications of replacement of equivalents \((RE)\) in the instance coordinates and in the scope coordinate, we obtain
		\[
		\Box(\varphi_1,\ldots,\varphi_n;\chi)\dashv\vdash_{\mathrm{PINL}}\Box(\theta_1,\ldots,\theta_n;\rho),
		\]
		and 
		\[
		\Diamond(\varphi_1,\ldots,\varphi_n;\chi)\dashv\vdash_{\mathrm{PINL}}\Diamond(\theta_1,\ldots,\theta_n;\rho).
		\]
		Now let $a=(a_1,a_2)\in W^c$. Since $a_1$ is a theory and is closed under provable consequence, we have 
		\[
		\Box(\varphi_1,\ldots,\varphi_n;\chi)\in a_1\iff\Box(\theta_1,\ldots,\theta_n;\rho)\in a_1.
		\]
		Similarly, since $a_2$ is a counter-theory and is downward closed under provable consequence, we have 
		\[
		\Diamond(\varphi_1,\ldots,\varphi_n;\chi)\in a_2\iff\Diamond(\theta_1,\ldots,\theta_n;\rho)\in a_2.
		\]
		This proves the corollary.
	\end{proof}
\end{cor}
Corollary \ref{IRC} ensures that a modal label given by a tuple of admissible positive opens is independent of the choice of formulas representing those opens. Hence the following open-labelled neighbourhood assignments are well-defined.
\begin{defn}[Canonical Bitopological PINL-space]\label{CBP}
	The canonical bitopological PINL-space is the structure
	\[
	\mathfrak X^{c}_{\mathrm{PINL}}
	=
	\bigl(
	W^{c},
	\tau^{+,c},
	\tau^{-,c},
	A^{c},
	\mathcal N^{\Box,c},
	\mathcal N^{\Diamond,c}
	\bigr),
	\]
	where
	\[
	A^{c}=\{U_{\varphi}:\varphi\in \mathcal{L}_{\mathrm{PINL}}\},
	\qquad
	U_{\varphi}
	=
	\{a=(a_{1},a_{2})\in W^{c}:\varphi\in a_{1}\},
	\]
	and \(\tau^{+,c}\) is the topology generated by \(A^{c}\). Similarly,
	\[
	B^{c}=\{V_{\varphi}:\varphi\in \mathcal{L}_{\mathrm{PINL}}\},
	\qquad
	V_{\varphi}
	=
	\{a=(a_{1},a_{2})\in W^{c}:\varphi\in a_{2}\},
	\]
	and \(\tau^{-,c}\) is the topology generated by \(B^{c}\). Thus $(W^{c},\tau^{+,c},	\tau^{-,c})$ is a bitopological space.\\
	
	\noindent For each \(a=(a_{1},a_{2})\in W^{c}\), define
	\[
	\begin{aligned}
		\mathcal N^{\Box,c}(a)
		:=
		\Bigl\{&
		\bigl(
		(U_{\psi_{1}},\ldots,U_{\psi_{n}};U_{\chi}),
		S
		\bigr):
		\\
		&
		\Box(\psi_{1},\ldots,\psi_{n};\chi)\in a_{1},
		\\
		&
		S\in\mathcal P_{\mathrm{fin}}(W^{c}),
		\\
		&
		S\subseteq U_{\chi},\quad
		S\cap U_{\psi_{i}}\neq\varnothing
		\quad\text{for every } i=1,\ldots,n
		\Bigr\},
	\end{aligned}
	\]
	and
	\[
	\begin{aligned}
		\mathcal N^{\Diamond,c}(a)
		:=
		\Bigl\{&
		\bigl(
		(U_{\varphi_{1}},\ldots,U_{\varphi_{n}};U_{\psi}),
		C
		\bigr):
		\\
		&
		\Diamond(\varphi_{1},\ldots,\varphi_{n};\psi)\in a_{2},
		\\
		&
		C\in\mathcal P_{\mathrm{fin}}(W^{c}),
		\\
		&
		C\cap U_{\psi}=\varnothing,\quad
		C\nsubseteq U_{\varphi_{i}}
		\quad\text{for every } i=1,\ldots,n
		\Bigr\}.
	\end{aligned}
	\]
	For \(n=0\), the instance conditions are understood vacuously.
\end{defn}
\subsection*{Induced instantial operations on admissible positive opens}
First, we define two maps $f_n^c$ and $g_n^c$ whose arguments are admissible positive opens from $A^c$ and whose values are subsets of $W^c$, using the open-labelled neighbourhood assignments of the canonical bitopological PINL-space. \\

\noindent Let $n\in\omega$ and let $D_1,\ldots,D_n,E\in A^c$. Define $f_n^c:(A^c)^n\times A^c\to\mathcal{P}(W^c)$ by 
\[
\begin{aligned}
	f_n^c(D_1,\ldots,D_n;E)
	=\bigl\{a\in W^c :\;&
	\text{there exists }S\in\mathcal P_{\mathrm{fin}}(W^c) \text{ such that}\\
	&
	\bigl((D_1,\ldots,D_n;E),S\bigr)\in\mathcal N^{\Box,c}(a)\bigr\}.
\end{aligned}
\]
Similarly, define $g_n^c:(A^c)^n\times A^c\to\mathcal{P}(W^c)$ by
\[
\begin{aligned}
	g_n^c(D_1,\ldots,D_n;E)
	=
	\bigl\{a\in W^c :\;&
	\text{there is no }C\in\mathcal P_{\mathrm{fin}}(W^c) \text{ such that}\\
	&
	\bigl((D_1,\ldots,D_n;E),C\bigr)\in\mathcal N^{\Diamond,c}(a)\bigr\}.
\end{aligned}
\]
Thus $(f_n^c)$ records the existence of a finite labelled witness set for the box-type operation, while $(g_n^c)$ records the absence of a finite labelled set failing the co-witness condition for the diamond-type operation.
\begin{prop}\label{MRAPO}
	Let $n\in\omega$. 
	\begin{enumerate}[(i)]
		\item For all formulas $\psi_1,\ldots,\psi_n,\chi\in\mathcal L_{\mathrm{PINL}}$,
		\[
		f_n^c(U_{\psi_1},\ldots,U_{\psi_n};U_\chi)=U_{\Box(\psi_1,\ldots,\psi_n;\chi)}.
		\]
		\item For all formulas $\varphi_1,\ldots,\varphi_n,\psi\in\mathcal L_{\mathrm{PINL}}$,
		\[
		g_n^c(U_{\varphi_1},\ldots,U_{\varphi_n};U_\psi)=U_{\Diamond(\varphi_1,\ldots,\varphi_n;\psi)}.
		\]
	\end{enumerate}
	\begin{proof}
		Let $\alpha=\Box(\psi_1,\ldots,\psi_n;\chi)$. We first prove that 
		\[
		U_\alpha= f_n^c(U_{\psi_1},\ldots,U_{\psi_n};U_\chi).
		\]
		Let $a=(a_1,a_2)\in U_\alpha$. Then
		\[
		\alpha\in a_1.
		\]
		By Lemma \ref{WA}, there exists \(S\in\mathcal P_{\mathrm{fin}}(W^c)\) such that
		\[
		S\subseteq U_\chi
		\]
		and 
		\[
		S\cap U_{\psi_i}\neq\varnothing\text{ for every } i=1,\ldots,n.
		\]
		Therefore, by the definition of \(\mathcal N^{\Box,c}(a)\),
		\[
		\bigl((U_{\psi_1},\ldots,U_{\psi_n};U_\chi),S\bigr)\in \mathcal{N}^{\Box,c}(a).
		\]
		Hence, 
		\[
		a\in f_n^c(U_{\psi_1},\ldots,U_{\psi_n};U_\chi).
		\]
		Conversely, suppose
		\[
		a\in f_n^c(U_{\psi_1},\ldots,U_{\psi_n};U_\chi).
		\]
		Then there exists a finite set $S\subseteq W^c$ such that 
		\[
		\bigl((U_{\psi_1},\ldots,U_{\psi_n};U_\chi),S\bigr)\in \mathcal{N}^{\Box,c}(a).
		\]
		By the definition of $\mathcal N^{\Box,c}(a)$, there exists formulas $\theta_1,\ldots,\theta_n,\rho\in\mathcal L_{\mathrm{PINL}}$ such that 
		\[
		\begin{aligned}
			U_{\theta_i}&=U_{\psi_i}\text{ for every } i=1,\ldots,n,\\
			U_\rho&=U_\chi,
		\end{aligned}
		\]
		and 
		\[
		\Box(\theta_1,\ldots,\theta_n;\rho)\in a_1.
		\]
		By Corollary \ref{IRC}, 
		\[
		\Box(\theta_1,\ldots,\theta_n;\rho)\dashv\vdash_{\mathrm{PINL}}\Box(\psi_1,\ldots,\psi_n;\chi).
		\]
		Since $a_1$ is theory and is closed under provable consequence, it follows that
		\[
		\Box(\psi_1,\ldots,\psi_n;\chi)\in a_1.
		\]
		Thus $a\in U_\alpha$. Hence $f_n^c(U_{\psi_1},\ldots,U_{\psi_n};U_\chi)=U_{\Box(\psi_1,\ldots,\psi_n;\chi)}$.\\
				  
		  \noindent Now we show that for all formulas $\varphi_1,\ldots,\varphi_n,\psi\in\mathcal L_{\mathrm{PINL}}$,
		  \[
		  g_n^c(U_{\varphi_1},\ldots,U_{\varphi_n};U_\psi)=U_{\Diamond(\varphi_1,\ldots,\varphi_n;\psi)}.
		  \]
	Let $\delta=\Diamond(\varphi_1,\ldots,\varphi_n;\psi)$. We first show that 
	\[
	U_\delta
	\subseteq
	g_n^c(U_{\varphi_1},\ldots,U_{\varphi_n};U_\psi).
	\]
	Let $a=(a_1,a_2)\in U_\delta$. Then $\delta\in a_1$. Assume, for contradiction, that 
	\[
	a\notin g_n^c(U_{\varphi_1},\ldots,U_{\varphi_n};U_\psi).
	\]
	Then by definition of $g_n^c$ there exists \(C\in\mathcal P_{\mathrm{fin}}(W^c)\) such that 
	\[
	\bigl((U_{\varphi_1},\ldots,U_{\varphi_n};U_\psi),C\bigr)
	\in\mathcal N^{\Diamond,c}(a).
	\]
	By the definition of $N^{\Diamond,c}(a)$, there exist formulas \(\theta_1,\ldots,\theta_n,\rho\in\mathcal{L}_{\mathrm{PINL}}\) such that 
	\[
	\begin{aligned}
	U_{\theta_i}&=U_{\varphi_i}
	\qquad\text{for every }i=1,\ldots,n,\\
	U_\rho&=U_\psi,
\end{aligned}
	\]
	and 
	\[
	\Diamond(\theta_1,\ldots,\theta_n;\rho)\in a_2.
	\]
	By Corollary \ref{IRC}, 
	\[
	\Diamond(\theta_1,\ldots,\theta_n;\rho)
	\dashv\vdash_{\mathrm{PINL}}
	\Diamond(\varphi_1,\ldots,\varphi_n;\psi).
	\]
	Since $a_2$ is a counter-theory and downward closed under provable consequence, it follows that
	\[
	\Diamond(\varphi_1,\ldots,\varphi_n;\psi)\in a_2.
	\]
	Thus $\delta\in a_1\cap a_2$, contradicting the disjointness of the prime theory pair $a=(a_1,a_2)$. Therefore, 
	\[
	a\in g_n^c(U_{\varphi_1},\ldots,U_{\varphi_n};U_\psi).
	\]
Conversely, suppose
	\[
	a\in g_n^c(U_{\varphi_1},\ldots,U_{\varphi_n};U_\psi).
	\]
	We claim that \(a\in U_\delta\), equivalently \(\delta\in a_1\). Suppose, for contradiction, that $\delta\notin a_1$. Since $a=(a_1,a_2)$ is a prime theory pair, we have $\delta\in a_2$. By Lemma \ref{CTNSL}, there exists a finite set \(C\subseteq W^c\) such that 
	\[
	C\cap U_\psi=\varnothing
	\] 
	and 
	\[
	C\nsubseteq U_{\varphi_i}
	\qquad\text{for every }i=1,\ldots,n.
	\]
 Therefore, by the definition of \(\mathcal N^{\Diamond,c}(a)\),
	\[
	\bigl((U_{\varphi_1},\ldots,U_{\varphi_n};U_\psi),C\bigr)
	\in\mathcal N_\Diamond^c(a).
	\]
	This contradicts
	\[
	a\in g_n^c(U_{\varphi_1},\ldots,U_{\varphi_n};U_\psi).
	\]
	Hence, $\delta\in a_1$, and consequently $a\in U_\delta$. \\
	\noindent Thus we have 
	\[
		g_n^c(U_{\varphi_1},\ldots,U_{\varphi_n};U_\psi)=U_{\Diamond(\varphi_1,\ldots,\varphi_n;\psi)}.
	\]
	\end{proof}
\end{prop}
Proposition \ref{MRAPO} shows that, for every $n\in \omega$, the maps $f_n^c$ and $g_n^c$ take values in $A^c$. Hence they determine $(n+1)$-ary operations 
\[
f_n^c,g_n^c:(A^c)^n\times A^c\longrightarrow A^c.
\]
Therefore by Definition \ref{TBOC}, the canonical bitopological PINL-space $\mathfrak X^{c}_{\mathrm{PINL}}$ is operation-closed. It remains to verify that these induced operations satisfy the defining laws of a $2$-$\mathrm{DLIO}$.
\begin{cor}\label{AACBPS}
	The canonical bitopological PINL-space 
	\[
	\mathfrak X^{c}_{\mathrm{PINL}}
	=
	\bigl(
	W^{c},
	\tau^{+,c},
	\tau^{-,c},
	A^{c},
	\mathcal N^{\Box,c},
	\mathcal N^{\Diamond,c}
	\bigr),
	\]
	is algebraically admissible. Consequently,
	\[
	\mathbb{A}(	\mathfrak X^{c}_{\mathrm{PINL}})=\bigl(
	A^c,\cap,\cup,\varnothing,W^c,
	(f_n^c)_{n\in\omega},
	(g_n^c)_{n\in\omega}
	\bigr)
	\]
	is a $2$-$\mathrm{DLIO}$.
	\begin{proof}
		By Note \ref{ACDLIO}, 
		\[
		(\mathcal A^c,\cap,\cup,\varnothing,W^c)
		\]
		is a bounded distributive lattice. By Proposition \ref{MRAPO}, the operations
		$f_n^c,g_n^c:(A^c)^n\times A^c\longrightarrow A^c$ are well-defined for every \(n\in\omega\). Hence the canonical bitopological \(\mathrm{PINL}\)-space is operation-closed. It remains to verify the defining \(2\)-\(\mathrm{DLIO}\) laws. Since every element of \(A^c\) is of the form \(U_\varphi\), each law follows from the corresponding modal axiom of \(\mathrm{PINL}\), together with Lemma \ref{SLAO} and Proposition \ref{MRAPO}.\\
		 For example, we verify the \((F6)\)-law. Let 	\[
		 D_i=U_{\varphi_i},\qquad E=U_\psi,\qquad G=U_\gamma,\qquad H=U_\delta .
		 \]
		 and suppose $G\cup H=W^c$. Then $U_\gamma\cup U_\delta=U_{\gamma\vee\delta}=U_\top$, and hence by Lemma \ref{SLAO}, 
		\[
		\gamma\vee\delta\dashv\vdash_{\mathrm{PINL}}\top.
		\]
		The modal axiom $(\Box 6)$ gives
		\[
		\Box(\varphi_1,\ldots,\varphi_n;\psi)
		\vdash_{\mathrm{PINL}}
		\Box(\varphi_1,\ldots,\varphi_n,\gamma;\psi)
		\vee
		\Box(\varphi_1,\ldots,\varphi_n;\psi\wedge\delta).
		\]
		Using Lemma \ref{SLAO}, we get
		\[
		U_{\Box(\varphi_1,\ldots,\varphi_n;\psi)}\subseteq U_{\Box(\varphi_1,\ldots,\varphi_n,\gamma;\psi)}\cup U_{	\Box(\varphi_1,\ldots,\varphi_n;\psi\wedge\delta)}.
		\]
		Using Proposition \ref{MRAPO}, we get 
		\[
		\begin{aligned}
			f_n^c(D_1,\ldots,D_n;E)
			\subseteq&
			f_{n+1}^c(D_1,\ldots,D_n,G;E)
			\cup
			f_n^c(D_1,\ldots,D_n;E\cap H).
		\end{aligned}
		\]
		Similarly, if
		\[
		G\cap H=\emptyset,
		\]
		then the \(\Diamond\)-cover axiom \((\Diamond6)\) gives
		\[
		g_{n+1}^c(D_1,\ldots,D_n,G;E)
		\cap
		g_n^c(D_1,\ldots,D_n;E\cup H)
		\subseteq
		g_n^c(D_1,\ldots,D_n;E).
		\]
		The remaining \(f_n^c\)- and \(g_n^c\)-laws follow in the same way from the corresponding \(\Box\)- and \(\Diamond\)-axioms of \(\mathrm{PINL}\). Therefore the induced operations satisfy all defining \(2\)-\(\mathrm{DLIO}\) laws. Hence \(\mathfrak X^c_{\mathrm{PINL}}\) is algebraically admissible, and by Proposition \ref{TBPA2}, $\mathbb A(\mathfrak X^c_{\mathrm{PINL}})$ is a \(2\)-\(\mathrm{DLIO}\).

	\end{proof}
\end{cor}
\begin{defn}\label{IBL}
	Let $\mathbb{A}=(A,\wedge_A,\vee_A,0_A,1_A,(f_n^A)_{n\in\omega},(g_n^A)_{n\in\omega})$ and $\mathbb{B}=(B,\wedge_B,\vee_B,0_B,1_B,(f_n^B)_{n\in\omega},(g_n^B)_{n\in\omega})$ be two $2$-$\mathrm{DLIO}$s. A map
	\[
	h:\mathbb A\to\mathbb B
	\]
	is a $2$-$\mathrm{DLIO}$ homomorphism, if for all \(x,y,x_1,\ldots,x_n,z\in A\) and every $n\in \omega$, the following conditions hold:
	\begin{enumerate}[(i)]
		\item \(h\) preserves finite meets:
		\[
		h(x\wedge_A y)=h(x)\wedge_B h(y);
		\]
		\item \(h\) preserves finite joins:
		\[
		h(x\vee_A y)=h(x)\vee_B h(y);
		\]
		\item \(h\) preserves the bounds:
		\[
		h(0_A)=0_B,
		\qquad
		h(1_A)=1_B;
		\]
		\item \(h\) preserves the $\Box$-type instantial operations:
		\[
		h(f_n^A(x_1,\ldots,x_n;z))
		=
		f_n^B(h(x_1),\ldots,h(x_n);h(z));
		\]
		\item \(h\) preserves the $\Diamond$-type instantial operations:
		\[
		h(g_n^A(x_1,\ldots,x_n;z))
		=
		g_n^B(h(x_1),\ldots,h(x_n);h(z)).
		\]
	\end{enumerate}
	A $2$-$\mathrm{DLIO}$ homomorphism $h:\mathbb A\to\mathbb B$ is called a $2$-$\mathrm{DLIO}$ isomorphism if \(h\) is bijective. In this case, \(\mathbb A\) and \(\mathbb B\) are said to be isomorphic.
\end{defn}
By Corollary \ref{AACBPS}, the admissible-open algebra \(\mathbb A(\mathfrak X_{\mathrm{PINL}}^c)\) is a 2-$\mathrm{DLIO}$. We now show that it is isomorphic to the Lindenbaum $2$-$\mathrm{DLIO}$ of $\mathrm{PINL}$ under the natural map sending each element of Lindenbaum algebra of PINL to its canonical positive open.
\begin{thm}\label{CRL2}
	Let $\mathfrak L_{\mathrm{PINL}}$ be the Lindenbaum $2$-$\mathrm{DLIO}$ of PINL. Define a map 
	\[
	\eta:\mathfrak L_{\mathrm{PINL}}\longrightarrow
	\mathbb A(\mathfrak X_{\mathrm{PINL}}^c)
	\]
	by $\eta([\varphi])=U_\varphi$. Then $\eta$ is an isomorphism of $2$-$\mathrm{DLIO}$s.
	\begin{proof}
		First, we show that $\eta$ is well-defined. Suppose $[\varphi]=[\psi]$. Then $\varphi\dashv\vdash_{\mathrm{PINL}}\psi$. By Lemma \ref{SLAO}, $U_\varphi=U_\psi$. Hence,
		\[
		\eta([\varphi])=\eta([\psi]).
		\]
		So, $\eta$ is well-defined. Next, we show that \(\eta\) preserves the bounded distributive lattice structure. We have
		\[
		\begin{aligned}
			\eta([\varphi]\wedge[\psi])
			&=U_{\varphi\wedge\psi}\\
			&=U_\varphi\cap U_\psi\\
			&=\eta([\varphi])\cap\eta([\psi]).
		\end{aligned}
		\]
		Similarly, $\eta([\varphi]\vee[\psi])=\eta([\varphi])\cup\eta([\psi])$. Also, $\eta([\bot])=U_\bot=\varnothing$ and $\eta([\top])=U_\top=W^c$.\\
		Now we show that $\eta$ preserves the $\Box$-type instantial operations. Let $n\in\omega$, and let
		\[
		[\varphi_1],\ldots,[\varphi_n],[\psi]\in \mathfrak L_{\mathrm{PINL}}.
		\] 
		By definition of $f_n^{\mathfrak L}$, $f_n^{\mathfrak L}([\varphi_1],\ldots, [\varphi_n]; [\psi]) = [\Box(\varphi_1,\ldots, \varphi_n; \psi)]$. Then
		\[
		\begin{aligned}
			\eta\bigl(f_n^{\mathfrak L}([\varphi_1],\ldots,[\varphi_n];[\psi])\bigr)
			=
			\eta([\Box(\varphi_1,\ldots,\varphi_n;\psi)])
			=
			U_{\Box(\varphi_1,\ldots,\varphi_n;\psi)}.
		\end{aligned}
		\]
		Now
		\[
		\begin{aligned}
			f_n^c(\eta[\varphi_1],\ldots,\eta([\varphi_n]);\eta([\psi]))
			&=f_n^c(U_{\varphi_1},\ldots,U_{\varphi_n};U_\psi).
		\end{aligned}
		\]
		By Proposition \ref{MRAPO}, 
		\[
		\begin{aligned}
		f_n^c(U_{\varphi_1},\ldots,U_{\varphi_n};U_\psi)
		&=U_{\Box(\varphi_1,\ldots,\varphi_n;\psi)}.
	\end{aligned}
		\]
		Thus
		\[
		\begin{aligned}
			\eta\bigl(f_n^{\mathfrak L}([\varphi_1],\ldots,[\varphi_n];[\psi])\bigr)
			&=	
			f_n^c\bigl(\eta[\varphi_1],\ldots,\eta([\varphi_n]);\eta([\psi])\bigr).
		\end{aligned}
		\]
		Therefore, $\eta$ preserves the operation $f_n^{\mathfrak L}$, for every $n\in \omega$.\\
		We next verify that $\eta$ preserves the $\Diamond$-type instantial operations. By definition of $g_n^{\mathfrak L}$, 
		\[
		g_n^{\mathfrak L}([\varphi_1],\ldots, [\varphi_n]; [\psi]) = [\Diamond(\varphi_1,\ldots, \varphi_n; \psi)].
		\]
		Therefore,
		\[
		\begin{aligned}
			\eta\bigl(g_n^{\mathfrak L}([\varphi_1],\ldots, [\varphi_n]; [\psi])\bigr)
			&=
			\eta([\Diamond(\varphi_1,\ldots, \varphi_n; \psi)])
			=
			U_{\Diamond(\varphi_1,\ldots,\varphi_n;\psi)}.
		\end{aligned}
		\]
		Now 
		\[
		\begin{aligned}
				g_n^c(\eta([\varphi_1]),\ldots,\eta([\varphi_n]);\eta([\psi]))
			&=g_n^c(U_{\varphi_1},\ldots,U_{\varphi_n};U_\psi).
		\end{aligned}
		\]
		By Proposition \ref{MRAPO}, 
			\[
		\begin{aligned}
			g_n^c(U_{\varphi_1},\ldots,U_{\varphi_n};U_\psi)
			&=U_{\Diamond(\varphi_1,\ldots,\varphi_n;\psi)}.
		\end{aligned}
		\]
		Thus
		\[
		\begin{aligned}
			\eta\bigl(g_n^{\mathfrak L}([\varphi_1],\ldots,[\varphi_n];[\psi])\bigr)
			&=	
			g_n^c\bigl(\eta[\varphi_1],\ldots,\eta([\varphi_n]);\eta([\psi])\bigr).
		\end{aligned}
		\]
		Therefore, $\eta$ preserves the $\Diamond$-type instantial operation $g_n^{\mathfrak L}$, for every $n\in \omega$.\\
		It remains to show that $\eta$ is bijective. To prove injectivity, suppose $\eta([\varphi])=\eta([\psi])$. Then $U_\varphi=U_\psi$. By Lemma \ref{SLAO}, 
		\[
		\varphi\dashv\vdash_{\mathrm{PINL}}\psi.
		\]
		Hence, $[\varphi]=[\psi]$. Thus $\eta$ is injective.\\
		To prove surjectivity, let $D\in A^c$. By definition of $A^c$, there exists a formula \(\varphi\in\mathcal L_{\mathrm{PINL}}\) such that 
		\[
		D=U_\varphi.
		\]
		Therefore, $D=\eta([\varphi])$. Hence $\eta$ is surjective. Thus \(\eta\) is a bijective \(2\)-\(\mathrm{DLIO}\)-homomorphism. Therefore, \(\eta\) is an isomorphism of \(2\)-$\mathrm{DLIO}$s. 
	\end{proof}
\end{thm}
The preceding theorem is a canonical admissible-open representation result. It shows that the Lindenbaum $2$-$\mathrm{DLIO}$ of PINL is recovered as the algebra of admissible positive opens of the canonical bitopological PINL-space. This may be viewed as a first step toward a duality theory for PINL.
\section{Conclusion and future work}\label{CAFW}
In this study, we have developed Positive Instantial Neighbourhood Logic with two primitive modalities, \(\Box\) and \(\Diamond\), over a distributive-lattice propositional base. We have introduced the syntax and proof system of the logic and proved soundness with respect to persistent two-sided neighbourhood semantics.
A direct canonical proof for the original persistent two-sided neighbourhood semantics is difficult, because neighbourhoods are not labelled by the modal formulas for which they are used. We have introduced an auxiliary typed persistent neighbourhood semantics, which yields the required truth lemma and typed completeness theorem.\\
We also formulated the algebraic semantics of PINL by means of \(2\)-\(\mathrm{DLIO}\)s. These algebras have a bounded distributive lattice reduct together with two families of instantial operations corresponding to the $\Box$- and \(\Diamond\)-modalities. We have proved that the Lindenbaum algebra of formulas modulo provable equivalence is a \(2\)-\(\mathrm{DLIO}\), and that the deductive system is algebraically sound and complete with respect to the class of \(2\)-\(\mathrm{DLIO}\)s.\\
Finally, we have constructed the canonical bitopological PINL-space associated with the prime-pair construction. In this space, the positive opens \(U_\varphi\) record membership of \(\varphi\) in the positive component of prime pairs, while the negative opens \(V_\varphi\) record membership of \(\varphi\) in the negative component. The admissible positive opens carry the algebraic operations induced by the typed instantial neighbourhood structure. We have proved that the algebra of admissible positive opens of the canonical bitopological PINL-space is isomorphic to the Lindenbaum \(2\)-$\mathrm{DLIO}$ of PINL. Thus the final topological result of the paper is a canonical admissible-open representation theorem.\\
\noindent There are several directions for future work. First, we can develop a Priestley-style or bitopological duality for $2$-$\mathrm{DLIO}$s. This direction is motivated by Priestley duality and its modal extensions \cite{priestley1970representation,celani1999priestley}, by the duality between bitopological spaces and $d$-frames \cite{jakl2018d}, and by the existing coalgebraic duality theory for INL \cite{bezhanishvili2020duality}. To develop such duality, it will be necessary to specify a suitable category of descriptive PINL-spaces, appropriate morphisms, and stability criteria ensuring that the instantial neighbourhood structure is represented functorially. \\
\noindent Second, one may study additional interaction axioms between the two modal operations $\Box$ and $\Diamond$, and compare the resulting algebras with the DLIO framework arising in the duality theory of INL \cite{bezhanishvili2020duality}.\\
Third, one may develop a geometric version of PINL. In that direction, the bounded distributive lattice of admissible positive opens should be replaced by a frame of opens, and the instantial operations should satisfy suitable Scott-continuity conditions. This would provide a basis for future work on the hypothesis that a geometric version of INL can be obtained from positive INL by requiring frame-theoretic continuity conditions.\\
Another direction, motivated by recent developments in inquisitive neighbourhood logic in  \cite{ciardelli2025inquisitive}, is to investigate whether PINL can be extended in an inquisitive setting. Inquisitive neighbourhood logic studies neighbourhood models with a language that can express not only statements but also questions. Since PINL is a positive instantial logic over neighbourhood structures, it would be interesting to study whether PINL can be extended in an inquisitive direction, so that its instantial modalities can also interact with questions or information states. It would also be useful to investigate whether the witness and co-witness conditions of PINL have a natural bisimulation-theoretic analysis.







\begin{thebibliography}{99}
	\bibitem{davey2002introduction} Davey, Brian A., Priestley, Hilary A.$\colon$ Introduction to lattices and order, CUP, 2002.
	\bibitem{kelly1963bitopological} Kelly, J. C.$\colon$ Bitopological Spaces, Proceedings of the London Mathematical Society, 3(1), 71-89, 1963.
	\bibitem{chellas1980modal} Chellas, Brian F.$\colon$ Modal logic: an introduction, CUP, 1980.
	\bibitem{van2017instantial} Van Benthem, Johan., Bezhanishvili, Nick.,Enqvist, Sebastian., Yu, Junhua.$\colon$ Instantial neighbourhood logic. The Review of Symbolic Logic, CUP, 10(1), 116-144, 2017.
	\bibitem{bezhanishvili2020duality} Bezhanishvili, Nick., Enqvist, Sebastian., De Groot, Jim.$\colon$ Duality for instantial neighbourhood logic via coalgebra. International Workshop on Coalgebraic Methods in Computer Science, Springer, 32-54, 2020.
	\bibitem{dunn1995positive} Dunn, J Michael.$\colon$ Positive modal logic, Studia Logica, 55(2), 301-317, 1995.
	\bibitem{priestley1970representation} Priestley, Hilary A.$\colon$ Representation of distributive lattices by means of ordered Stone spaces, Bulletin of the London Mathematical Society, 2(2), 186-190, 1970.
	\bibitem{celani1999priestley} Celani, Sergio., Jansana, Ramon. $\colon$ Priestley duality, a Sahlqvist theorem and a Goldblatt-Thomason theorem for positive modal logic, Logic Journal of the IGPL, OUP, 7(6), 683-715, 1999.
	\bibitem{jakl2018d} Jakl, Tom{\'a}{\v{s}}.$\colon$ d-Frames as algebraic duals of bitopological spaces, PhD thesis, Univerzita Karlova, Matematicko-fyzik{\'a}ln{\'\i} fakulta., 2018.
	\bibitem{celani1997new} Celani, Sergio., Jansana, Ramon.$\colon$ A new semantics for positive modal logic, Notre Dame Journal of Formal Logic, 38(1), 1-18, 1997.
	\bibitem{celani2012note} Celani, Sergio., Jansana, Ramon.$\colon$ A note on the model theory for positive modal logic, Fundamenta Informaticae, 114(1), 31-54, 2012.
	\bibitem{gehrke2005sahlqvist} Gehrke, Mai., Nagahashi, Hideo., Venema, Yde.$\colon$ A Sahlqvist theorem for distributive modal logic, Annals of pure and applied logic, 131(1-3), 65-102, 2005.
	\bibitem{kikot2018strictly} Kikot, Stanislav., Kurucz, Agi., Wolter, Frank., Zakharyaschev, Michael.$\colon$ On Strictly Positive Modal Logics with S4. 3 Frames., Advances in Modal Logic, 12, 399-418, 2018.
	\bibitem{palmigiano2004coalgebraic} Palmigiano, Alessandra.$\colon$ A coalgebraic view on positive modal logic, 327(1-2), 175-195, 2004.
	\bibitem{sadrzadeh2010positive} Sadrzadeh, Mehrnoosh., Dyckhoff, Roy.$\colon$ Positive logic with adjoint modalities: Proof theory, semantics, and reasoning about information, The Review of Symbolic Logic, 3(3), 351-373, 2010.
	\bibitem{de2021positive} De Groot, Jim.$\colon$ Positive monotone modal logic, Studia Logica, 109(4), 829-857, 2021.
	\bibitem{bezhanishvili2022coalgebraic} Bezhanishvili, Nick., De Groot, Jim., Venema, Yde.$\colon$ Coalgebraic geometric logic: basic theory, Logical Methods in Computer Science, 18, 2022.
	\bibitem{scott1970advice} Scott, Dana.$\colon$ Philosophical problems in logic: Some recent developments, 143-173, 1970.
	\bibitem{montague1970universal} Montague, Richard.$\colon$ Universal grammar, Theoria, 36, 373-398, 1970.
	\bibitem{hansen2009neighbourhood} Hansen, Helle Hvid., Kupke, Clemens., Pacuit, Eric.$\colon$ Neighbourhood structures: Bisimilarity and basic model theory, Logical Methods in Computer Science, 5(2), 2009.
	\bibitem{pacuit2017neighborhood} Pacuit, Eric.$\colon$ Neighborhood semantics for modal logic, Springer, 2017.
	\bibitem{ciardelli2025inquisitive} Ciardelli, Ivano.$\colon$ Inquisitive Neighborhood Logic, Journal of Logic, Language and Information, 34(5), 419-461, 2025.
	\bibitem{vickers2004double} Vickers, Steven.$\colon$ The double powerlocale and exponentiation: a case study in geometric logic, Theory and Applications of Categories, 12(13), 372-422, 2004.
	\bibitem{vickers1989topology} Vickers, Steven.$\colon$ Topology via logic, Cambridge University Press, 1989.
\end{thebibliography}
\end{document}